\documentclass{article}
\usepackage[utf8]{inputenc}
\usepackage{authblk}
\usepackage{setspace}
\usepackage[margin=1.25in]{geometry}
\usepackage{lmodern}
\usepackage{subfigure}% Support for small, `sub' figures and tables

\usepackage[numbers]{natbib}% Citation support using natbib.sty

\usepackage{amsmath,amssymb}
\usepackage{amsthm}

\theoremstyle{plain}% Theorem-like structures provided by amsthm.sty
\newtheorem{theorem}{Theorem}[section]
\newtheorem{lemma}[theorem]{Lemma}

\newtheorem{assumption}{Assumption}

\usepackage{multirow}
\usepackage{graphicx}
\usepackage{enumerate}
\usepackage{latexsym}
\usepackage{secdot}
\usepackage{url}
\usepackage{color}
\usepackage{fancyhdr}
\usepackage{mathrsfs}
\usepackage{booktabs}
\usepackage{amsfonts}
\usepackage{euscript}
\usepackage[final]{hyperref} % adds hyper links inside the generated pdf file
\usepackage{soul}
\usepackage{bm}
\usepackage{algorithm,algpseudocode}
\makeatletter
\newcommand*{\rom}[1]{\expandafter\@slowromancap\romannumeral #1@}
\makeatother

\newcommand{\E}{{\mathbb E}}
\newcommand{\bbP}{{\mathbb P}}
\newcommand{\var}{{\operatorname{Var}}}
\newcommand{\cov}{{\operatorname{Cov}}}

\newcommand{\tr}{{\operatorname{tr}}}
\newcommand{\ARL}{{\rm ARL}}
\newcommand{\beq}{\begin{equation}}
\newcommand{\eeq}{\end{equation}}
\newcommand{\Span}{\text{span}}
\newcommand{\Rank}{\text{Rank}}
\newcommand{\OO}{\operatorname{O}}

\newcommand{\tilSig}{\widetilde{\Sigma}}
\newcommand{\tilP}{\widetilde{P}}

\usepackage{xcolor}
\newcommand{\floor}[1]{\mbox{$\lfloor{#1}\rfloor$}}
\newcommand{\sqbrack}[1]{\mbox{$\left[{#1}\right]$}}
\newcommand{\cirbrack}[1]{\mbox{$\left({#1}\right)$}}
\title{Distribution-Free Online Change Detection for Low-Rank Images}

\author[1]{Tingnan Gong}
\author[1]{Seong-Hee Kim}
\author[1]{Yao Xie}

\affil[1]{H. Milton Stewart School of Industrial and Systems Engineering,\\
Georgia Institute of Technology}

\date{}

\begin{document}

\maketitle

%%%%%% Abstract %%%%%%
\begin{abstract}
We present a distribution-free CUSUM procedure designed for online change detection in a time series of low-rank images, particularly when the change causes a mean shift. We represent images as matrix data and allow for temporal dependence, in addition to inherent spatial dependence, before and after the change. The marginal distributions are assumed to be general, not limited to any specific parametric distribution. We propose new monitoring statistics that utilize the low-rank structure of the in-control mean matrix. Additionally, we study the properties of the proposed detection procedure, assessing whether the monitoring statistics effectively capture a mean shift and evaluating the rate of increase in the average run length relative to the control limit in both the in-control and out-of-control cases. The effectiveness of our procedure is demonstrated through simulated and real data experiments.

\end{abstract}

%%%%%% Main Text %%%%%%

\section{Introduction} \label{sec:intro}
In modern manufacturing, rapid improvements in sensor technology allow the industry to acquire data with dimensions much higher than those of decades ago. For example, images of layers in 3D printing \citep{Caltanissetta2023} can be obtained approximately every five seconds. For monitoring purposes, the image data can be viewed as matrix data. However, online monitoring of matrix data faces  multifaceted challenges:
\begin{itemize}
    \item \textit{Temporal independence:} Due to the high sampling rate, weak autocorrelation among consecutive images is nearly inevitable. Long-lasting autocorrelation is also common in specific industrial problems. For example, in the case of in-situ detection of laser power bed fusion (LPBF) \citep{gibson2021additive}, an anomaly called a hot spot may occur. When the laser beam repeatedly irradiates a thermally insulated region, heat builds up exceptionally quickly, leading to the formation of a hot spot. In the view of a thermographic camera, the pixels in the center of the overheating area remain hot (high intensity) while the edges slowly cool down, causing durable auto-correlation and nonstationary spatial correlation. 
    
    \item  \textit{Spatial independence: }
    Individual components of a high-resolution matrix tend to exhibit weak spatial correlation, thus violating the assumption of spatial independence. As mentioned in the aforementioned in-situ detection problem of LPBF, the hot spot anomaly can lead to unusual spatial correlation. 

    \item \textit{Data normality:} Violations of normality in data can result in a deviation from theoretical assumptions made by methods with the normality assumption, as discussed in \citep{yourstone1992non}.  \cite{chou1998transforming} tackle the normality issue by transforming low-dimensional vectors into nearly normally distributed data. However, this technique turns out to be computationally prohibitive in high dimensions. 
    % Another direction is to come up with distribution-free techniques \citep{perry2022distribution}. 

    \item \textit{High dimensionality:} Matrix data with large dimensions can drastically increase the complexity of practical algorithm. Additionally, rank deficiency is common which may cause performance degradation if not properly handled, which calls for algorithms that are designed for low-rank matrices. Many techniques, combined with dimensionality reduction techniques such as the principal components analysis (PCA) \citep{de2015overview} and process projection and fusion \citep{zhang2020spatial} may improve computational and statistical efficiency of the methods. 
\end{itemize}

 %Thus, the high-dimensional matrix data requires advanced technology beyond traditional assumptions to exploit and reinterpret. 

One way to handle high-dimensional matrix data is to adopt the profile monitoring perspective, as in \cite{woodall2004using}, \cite{woodall2007current} and \cite{maleki2018overview}. Procedures that employ this approach make a compromise by treating high-dimensional matrix data as a long vector of predictors, sacrificing the spatial data structure. Practitioners construct the response in single or multiple channels and develop appropriate models to characterize the functional relationship between the response and the predictors. The parametric regression model serves as one of the primary instruments in profile monitoring to describe this functional relationship. \cite{chang2006monitoring}, \cite{zou2007monitoring}, \cite{zhu2009monitoring} and \cite{noorossana2004monitoring} perform profile monitoring using linear regression models. \cite{kazemzadeh2009monitoring}, \cite{williams2007statistical}, and \cite{moguerza2007monitoring} introduce nonlinear profile monitoring to enhance interpretability. Other profile monitoring works utilize dimension reduction techniques to extract features, such as PCA and independent component analysis (ICA). \cite{colosimo2010comparison} discuss comparisons between regression-based and PCA-based profile monitoring procedures.

Another approach is to use wavelet-based distribution-free profile monitoring procedures \citep{lee2012monitoring,wang2015monitoring} for high-dimensional vectors. These methods can partially address the mentioned challenges, including general marginal distributions and high dimensionality. However, profile monitoring tends to alter the original matrix structure and results in a loss of spatial information. Additionally, none of the aforementioned works incorporates general temporal correlation.

To mitigate the loss of information caused by data structure transformation, some works exploit matrix characteristics directly and construct matrix-based monitoring procedures. \cite{megahed2011review} categorize various matrix-based monitoring procedures, including spatial and multivariate-matrix-analysis-based control charts.

Recent works either combine multiple popular techniques or develop new models on matrix data from images. Among the procedures that integrate multiple popular techniques, \cite{yan2014image} employ low-rank tensor decomposition to achieve dimension reduction and feature extraction and monitor the extracted features using Hotelling $T^2$ and Q control charts; the method does not strictly require data normality and independence. However, their Phase-\rom 1 calibration of the control limits is time-consuming due to the estimation of the empirical distribution of statistics. \cite{koosha2017statistical} and \cite{eslami2021spatial} utilize wavelet transformations on matrix data to extract features and then build control charts. \cite{koosha2017statistical} perform a generalized likelihood ratio (GLR) control chart on extracted features. \cite{eslami2021spatial} further extend the approach in \cite{koosha2017statistical} by combining it with a regression-based parametric model to accommodate underlying data correlation. However, they still require data normality. \cite{alaeddini2018spatiotemporal} compress data using partial least squares discriminant analysis and then construct the control statistics using Delaunay triangulation \citep{lee1980two} to segment the squared error matrix into triangles, computing the area within each triangle as control statistics. Despite the novelty of this technique, the assumption of data independence limits its broad application.  Region of interest (ROI) is a popular data compression technique for matrix processing. Based on this tool, \cite{megahed2012spatiotemporal} incorporate the GLR control chart to monitor the average intensity vector calculated from these regions, assuming the presence of at most one cluster of defects in the images.  \cite{he2016image} extend this approach further to detect multiple clusters of defects in images. \cite{amirkhani2020novel} consider the combination of ROIs and one-way analysis of variance (ANOVA).

 On the other hand, some procedures develop new models on the matrices. For a series of images with a smooth background, \cite{yan2017anomaly} propose a smooth-sparse decomposition (SSD) model to decompose observations into the background and potential sparse anomalies. However, the SSD model cannot incorporate temporal information into the data. Building on this work, \cite{yan2018real} extend the SSD methodology with a spatio-temporal smooth-sparse decomposition (ST-SSD) model to tackle temporal information and spatial patterns. 

Although initially not developed for matrix monitoring, some spatio-temporal monitoring charts \citep{jiang2011spatiotemporal, lee2014spatiotemporal, lee2015robust} used in the environment and public health surveillance can be adapted for matrix monitoring. Another thread of research addresses specific matrix monitoring problems in the industry. One of the applications that has garnered significant attention is matrix monitoring in metal additive manufacturing \citep{yan2022real, colosimo2018spatially, khanzadeh2019situ, liu2021integrated, guo2020hierarchical}.

% \st{a distribution-free procedure, namely} 

In this paper, we propose the distribution-free low-rank image monitoring (DFLIM) procedure to perform online change detection for a time series of matrices. This work extends our earlier work \cite{gong2024distribution}, which considers a special case when the in-control mean image matrix is rank-one; here, we consider the general case where in-control mean image matrix can have a rank higher than one but still low rank. Our procedure's monitoring statistics are constructed based on singular value decomposition (SVD) of the in-control mean matrix and projected observations. These statistics are then used in the CUSUM recursion, which can be computed recursively online. We analyze the theoretical properties of the DFLIM procedure in terms of the in-control average run length ($\ARL_0$) and out-of-control average run length ($\ARL_1$). These metrics evaluate how frequently a false alarm occurs when a monitored process is in control and how quickly a change can be detected when the process is out of control.  Furthermore, we empirically study the capability of the procedure against temporal dependence and non-normality of the data. The effectiveness of the proposed procedure is demonstrated through simulated and real data experiments.

The remainder of the paper is organized as follows: Section \ref{sec:setup} describes our problem and assumptions. Section \ref{sec:procedure} proposes the DFLIM procedure, which utilizes a CUSUM chart on  $T^2$ statistics extracted from the high-dimensional matrix data. In Section \ref{sec:analysis}, we conduct theoretical analysis on the mean shift size in our statistics and study the ARL behavior of the DFLIM procedure. In Section \ref{sec:simulation}, we design simulated experiments designed to demonstrate the performance of the DFLIM procedure and empirically support some claims that are challenging to prove analytically. Section \ref{sec:real} applies the DFLIM procedure to real data sets to demonstrate its broad applications, followed by concluding remarks in Section \ref{sec:conclusion}.

\section{Problem Setup}\label{sec:setup}

At each time $t$, we observe a matrix denoted as $X_t \in \mathbb R^{p_1\times p_2}$. At an unknown change point $k\in\{1,\ldots,\infty\}$, the mean of $X_t$ shifts from $M_0$ to $M_1$. We can formulate the problem in terms of online hypothesis testing:
\begin{equation}\label{eq:hypotheses}
\begin{split}
\text{H}_0:&\ X_t = M_0 +\epsilon_t,\quad t=1,2,\ldots \\
\text{H}_1:&\ X_t = 
\begin{cases}
M_0 + \epsilon_t,\quad t=1,2,\ldots,k-1, \\
M_1 + \epsilon_t,\quad t=k,k+1,\ldots.
\end{cases}
\end{split}
\end{equation}
where $M_0$ and $M_1$ are $p_1$-by-$p_2$ matrices representing the in-control and out-of-control mean matrices, respectively. We assume that $p_1 \le p_2$, which can always be achieved by transposing the image matrices if necessary. 
The noise matrix {$\epsilon_t\in \mathbb R^{p_1\times p_2}$} is assumed to have the same marginal distribution across time $t$, whereas the parametric form of the distribution is unrestricted, allowing for {spatial correlation within  $\epsilon_t$ for each $t$ and temporal correlations across $\epsilon_t$ across time.} We denote expectations of the observations before and after the change as $\mathbb E_0[\cdot]$ and $\mathbb E_1[\cdot]$, respectively; The covariance functions $\cov_0(\cdot)$ and $\cov_1(\cdot)$ are defined similarly. {Note that in \eqref{eq:hypotheses}, all the randomness are due to $\epsilon_t$; however, the observations $X_t$ have different mean before and after the change.}
%We slightly abuse the subscripts 0 and 1 in the expectation and covariance functions to indicate the probability measures corresponding to the distribution of $X_t$ with expectations $M_0$ and $M_1$, respectively.}

Before the change,  $X_t$ follows marginal distribution $F_0$ with mean 
\[\mathbb E_0[X_t] = M_0,\] where we assume $M_0$ to be known; this assumption is a reasonable simplification, since usually have in-control data to estimate the pre-change parameters accurately. For example, in a printed circuit board (PCB) manufacturing process, the in-control pattern is designed before production, or sufficient in-control data is available to accurately estimate the process parameters (further discussion can be found in Section~\ref{subsec:parameter}). 
% Note that we allow for temporal dependence among observations $X_t$ and non-normality of marginal distributions.

After the change (i.e., $t \ge k$), the observations $X_t$ experience a mean shift 
\[\mathbb E_1 [X_t] = M_1 :=  M_0 + A,\] where $A \in\mathbb R^{p_1\times p_2}$ represents the unknown non-zero {and deterministic} mean shift. Throughout the paper, we assume that the following assumptions hold:
\begin{assumption}\label{assump:dist}
The marginal distributions of $X_t$ for $t=1,2,\ldots, k-1$ and $X_t -A$ for $t=k, k+1, \ldots$ are $F_0$; i.e., the marginal distribution of $\epsilon_t$ stays the same for all $t$.     
\end{assumption}

\begin{assumption}\label{assump:low_rank}
The in-control mean $M_0$ is low-rank, with rank $r \ll \min\{p_1,p_2\}$, and its SVD is given as follows:
$$
M_0 = \sum_{i=1}^r \lambda_i  u_i  v_i^\top
$$
where $\lambda_i$, $u_i$, and  $v_i$ represent the singular values, left singular vectors, and right singular vectors, respectively.
\end{assumption}
%Assumption~\ref{assump:dist} implies that the shift affects the mean but that the marginal distributions remain unchanged after the shift. Assumption~\ref{assump:low_rank} ensures that the in-control mean matrix $M_0$ is low rank. 

\section{Distribution-Free Low-Rank Image Monitoring}\label{sec:procedure}

Given the matrix data $X_t$, we extract two types of projections for detection statistics. First, using the left and right singular vectors, $u_i$ and $v_i$, of the in-control mean $M_0$, we compute the {\it first-type} projected observations $\beta_{i,t}$ for $i=1,\ldots, r$ as follows:
\begin{equation}\label{formula:features coordinates}
 \beta_{i,t} =  u_i^\top  X_t  v_i.
\end{equation}
Next, we calculate the residual matrix 
\[
{R}_t =  X_t- M_0,
\]
and perform its SVD to obtain its singular values in descending order. We use the first $r$ singular values as the {\it second-type} projected observations $\gamma_{i, t}$ for $i=1,\ldots, r$.

The statistic consisting of both types of projected observations at time $t$ is denoted as $y_t$:
\begin{equation}\label{formula:yt}
     y_t = \left[\beta_t^\top, \gamma_t^\top\right]^\top =  \left[\beta_{1,t},\ldots,\beta_{r,t},\gamma_{1,t},\ldots,\gamma_{r,t}\right]^\top\in\mathbb R^{2r},
\end{equation}
where $\beta_t = \left[\beta_{1,t},\ldots,\beta_{r,t}\right]^\top$ and $\gamma_t=\left[\gamma_{1,t},\ldots,\gamma_{r,t}\right]^\top$.
Note that $\beta_{i,t}$ is the projection of $X_t$ onto the static directions $u_i$ and $v_i$, whereas $\gamma_{i,t}$ represents the projection onto temporally varying directions, i.e., the singular vectors of $R_t$. We discuss why it is necessary to incorporate both types of projections into $y_t$ to ensure detection power in Section \ref{sec:analysis}.

To construct a CUSUM procedure, we compute a $T^2$-type statistics $T_t$ using $y_t$: 
\begin{equation}
\label{formula:T2-like series}
T_t = \left( y_t - \E_0 [y_t] \right)^\top  \cov_0^{-1}(y_t) \left( y_t - \E_0[y_t] \right),
\end{equation}
which will be used in the increment in the CUSUM statistics. Note that $\cov_0(y_t)$ is shorthanded for $\cov_0(y_t, y_t)$, the covariance matrix of the vector $y_t$. %Throughout the paper, we use this abbreviated notation for the covariance matrix between the same vector. 

The CUSUM statistics for the DFLIM procedure are defined recursively with $S_0 = 0$:
\begin{equation}\label{formula:dist_free_cusum}
            S_t = \max\left\{0, S_{t-1} + \left(T_t - \E_0[ T_t] - c\sigma_T\right) \right\}, \quad t=1,\ldots,
\end{equation}
where $c$ is a pre-selected constant, and $\sigma_T>0$ is the in-control marginal standard deviation of $T_t$ (under the marginal distribution $F_0$). The DFLIM procedure stops and raises an out-of-control signal at time $\tau$  when the monitoring statistic exceeds a control limit $H >0$:
\begin{equation}
    \tau= \inf\{t>0: S_{t}\geq H\}. 
\end{equation}
Given the stopping time $\tau$, we define $\ARL_0=\E_0[\tau]$ and $\ARL_1=\E_1[\tau]$. Here, $\ARL_0$ represents the average time to raise a false alarm when the process is in-control. %, which is inversely proportional to the false alarm rate. 
$\ARL_1$ measures how quickly a monitoring procedure can detect a change when the process is  out-of-control.

The detailed description of the proposed DFLIM procedure is given in Algorithm \ref{algorithm:Distribution-Free CUSUM}. Note that in the inputs of Algorithm~\ref{algorithm:Distribution-Free CUSUM}, the in-control parameters, including $M_0$, $\E_0[y_t]$, and $\cov_0(y_t)$, are assumed to be known for theoretical analysis. In practice, the in-control parameters are estimated from in-control data, as described in Section~\ref{subsec:parameter} and Algorithm~\ref{algorithm:calibration on control limit}.

\begin{algorithm}[H]
\caption{Distribution-free Low-rank Image Monitoring}
\label{algorithm:Distribution-Free CUSUM}
\textbf{Input:} Sequence of observations $\{X_t,t = 1,2,\ldots\}$, in-control mean matrix $M_0$, in-control mean vector $\E_0[y_t]$, in-control covariance $\cov_0(y_t)$, constant $c$, standard deviation $\sigma_T$, and control limit $H$ (whose determination is discussed in Section~\ref{subsec:H}).\\
\textbf{Output:}  stopping time $\tau$. 
\begin{algorithmic}[1]
\State Initialize $\tau = +\infty$, $t = 0$, and $S_0 = 0$. 
\State Perform SVD on $M_0$ to obtain directions $u_i$ and $v_i$ for projections with $i=1,\ldots,r$. 

\While{$S_t<H$}
    \State Set $t=t+1$ and obtain $X_t$.
    \State Compute $\beta_{i,t}$ for $i=1,\ldots,r$ as in \eqref{formula:features coordinates}. 
    \State Obtain $r$ largest singular values of $R_t=X_t-M_0$, i.e.
    $\gamma_{i,t}$ for $i=1,\ldots,r$.
    \State Form $y_t$ as in \eqref{formula:yt} and compute $T_t$ as in \eqref{formula:T2-like series}. 
    \State Update the monitoring statistic $S_t$ as in \eqref{formula:dist_free_cusum}. 
\EndWhile
 \State Set $\tau = t$ and raise an out-of-control alarm. 
 % \st{$\tau = t$.  }
\end{algorithmic}
\end{algorithm}

\subsection{Control limit determination and setup phase} \label{subsec:H}

In this section, we derive an expression for ARLs of the DFLIM procedure and explain how this expression can be used to determine the control limit $H$. Additionally, we explain how to estimate the parameters necessary for implementing the DFLIM procedure.

Define the mean and variance parameters of $T_t$ as follows:
\begin{equation*}\label{formula:mean and variance of T2}
    m := \E [T_t], \quad \Omega^2 := \lim _{t \rightarrow \infty} t \var \left(\frac{\sum_{\ell=1}^t T_\ell}{t}\right)= \sum_{t = -\infty}^\infty \cov(T_0, T_t),
\end{equation*}
where the expectation and the variance can be taken under either in-control or out-of-control phase. Here, $\Omega^2$ represents the limiting variance parameter for $T_t$, which provides a better measure of process variability compared to marginal variance in the presence of temporal correlations. We can then define the standardized time-series of the first $t$ observations $\{T_1, T_2, \ldots, T_t\}$ as follows: 
\begin{equation*}
\mathcal{C}_t(s) \equiv \frac{\sum_{\ell=1}^{\lfloor t s\rfloor} T_\ell-t s m}{ \sqrt{t\Omega^2}}, \quad s \in[0,1]. 
\end{equation*}
In addition to Assumptions~\ref{assump:dist} and \ref{assump:low_rank}, we assume that $\{T_t : t=1,2,\ldots\}$ satisfies the Functional Central Limit Theorem (FCLT): 
\begin{assumption}[FCLT]
\label{asum:FCLT}
Given $\left\{T_t: t=1,2, \ldots\right\}$, the standardized time-series process $\mathcal{C}_t(\cdot)$ satisfies:
\begin{equation*}
    \mathcal{C}_t(\cdot) {\stackrel{\mathcal D}{\longrightarrow}} \mathcal{W}(\cdot)\quad\text{as}\quad t\to\infty,
\end{equation*}
in the space $D[0,1]$ where $\stackrel{\mathcal D}{\to}$ denotes convergence in distribution, $\mathcal W(\cdot)$ denotes a standard Brownian process, and the space $D[0,1]$ contains functions defined on $[0,1]$ that are right-continuous with left-hand limits.
\end{assumption}
For conditions under which the FCLT is applicable, one can refer to \cite{glynn1985large}. In Chapter 4.4 of \cite{whitt2002stochastic}, it is suggested that, in practical terms, it is generally justifiable to presume the validity of the FCLT when $\Omega^2$ is finite.

Under the FCLT, \cite{kim2007distribution} show that the limiting process of CUSUM statistics is closely related to a standard Brownian motion process, regardless of the parametric form of $X_t$, forming the basis for a distribution-free procedure. Furthermore, they derive an approximate expression of ARLs, including both $\ARL_0$ and $\ARL_1$, based on the properties of the converged process of CUSUM statistics. The DFLIM procedure also achieves the distribution-free property under the FCLT and determines the control limit $H$ using the approximate expression of ARLs of the converged process, as provided in Lemma~\ref{lem:ARL}. Similar results were reported in \cite{bagshaw1975effect}, however without a clear characterization of the process variability using $\Omega^2$.

\begin{lemma}[\cite{kim2007distribution}]
\label{lem:ARL}
If $\left\{T_t: t=1,2, \ldots\right\}$ satisfies Assumption \ref{asum:FCLT}, then
\begin{equation*}
\begin{aligned}
& \ARL \approx 
\begin{cases}
H^2 / \Omega^2, & \text { if } d_T=0, \\
\frac{\Omega^2}{2 d_T^2}\left[\exp \left(-\frac{2Hd_T}{\Omega^2}\right)-1+\frac{2Hd_T}{\Omega^2}\right], & \text { otherwise},
\end{cases}
\end{aligned}
\end{equation*}
for large $H$, where $d_T = \E[T_t] - \E_0[ T_t] - c\sigma_T$.
\end{lemma}
When the monitored process is in-control, we have $d_T = -c\sigma_T$. By plugging $d_T= -c\sigma_T$ in the ARL expression in Lemma~\ref{lem:ARL} and setting it equal to a target $\ARL_0$, one can analytically solve for the value of the control limit $H$. As suggested in \cite{siegmund1985sequential}, a more accurate control limit $H$ can be obtained by solving the following equation:
\begin{equation}
\label{formula:control limit}
\ARL_0 = \frac{\Omega_0^2}{2 c\sigma_T^2}\left\{\exp \left[\frac{2 c\sigma_T(H+1.166 \Omega_0)}{\Omega_0^2}\right]-1-\frac{2 c\sigma_T(H+1.166 \Omega_0)}{\Omega_0^2}\right\},
\end{equation}
where $\Omega^2_0$ is the in-control $\Omega^2$. More details on the calculation of $\Omega^2_0$ are in Appendix~\ref{sec:CvM}. We summarize the calculation of parameters and the control limit $H$, along with other parameters required for the implementation of the DFLIM procedure, in Algorithm \ref{algorithm:calibration on control limit}.

\begin{algorithm}[t!]
\caption{Setup phase for the DFLIM procedure}
\label{algorithm:calibration on control limit}
\textbf{Input:} Sequence of in-control observations $\{X_t,t = 1,2,\ldots,n\}$, in-control mean matrix $M_0$, rank $r$, constant $c$, and target $\ARL_0$.   \\
\textbf{Output:} $\E_0[y_t]$, $\cov_0(y_t)$, $\sigma_T$, and $H$.
 \begin{algorithmic}[1]
    \State Perform SVD on $M_0$ to obtain directions $u_i$ and $v_i$ for projections with $i=1,\ldots,r$.
    % \tn{I think it is fine given that they are the same.} 
    \While{$t<n$}
        \State Set $t=t+1$ and obtain $X_t$. 
        \State Compute 
        % \st{static-type statistics }
        $\beta_{i,t}$ for $i=1,\ldots,r$ as in \eqref{formula:features coordinates}. 
        \State Obtain $r$ largest singular values of $R_t=X_t-M_0$, i.e.
    $\gamma_{i,t}$ for $i=1,\ldots,r$.
        \State Form $y_t$ as in \eqref{formula:yt}. 
    \EndWhile
    \State Approximate $\E_0[y_t]$ by $\bar y = 1/n\sum_{t=1}^n y_t$ and  $\cov_0(y_t)$ by $1/(n-1)\sum_{t=1}^n(y_t-\bar y)(y_t-\bar y)^\top$. 
    \State Compute $T_t$ as in \eqref{formula:T2-like series} for $t=1,\ldots,n$.
    \State Compute $\bar T = 1/n\sum_{t=1}^n T_t$.
    \State Approximate $\sigma_T$ by $\sqrt{1/(n-1)\sum_{t=1}^n (\bar T-T_t)^2}$.
    \State Estimate variance parameter $\Omega_0^2$ of $T_t$ using Cram\'er-von Mises (CvM) estimator. 
    \State Calculate control limit $H$ by solving equation \eqref{formula:control limit}. 
 \end{algorithmic}
\end{algorithm}

 \subsection{Additional parameter calibration}
 \label{subsec:parameter}

In this subsection, we discuss how to estimate $M_0$ when the target mean image is not well-defined and how to determine $r$ from the estimated $M_0$. Additionally, we provide a recommendation for choosing a constant $c$.\\
 
\noindent{\it Estimation of $M_0$.} With training in-control data $X_1,\ldots,X_{n}$, we use an estimator for $M_0$ being the sample average of in-control observations, namely $\widehat{M}_0 = \sum_{t=1}^{n} X_t/n$. Due to the Law of Large Numbers, $\widehat{M}_0$  approximates the true in-control mean matrix $M_0$ well when $n$ is large. 
\\

\noindent{\it Selection of $r$.} The rank $r$ of $M_0$ is obtained numerically using singular value hard thresholding \citep{gavish2017optimal}. Specifically, for the singular values of $M_0$, denoted as $\lambda_1\ge \ldots \ge \lambda_{\min\{p_1,p_2\}}$, and given a threshold $\lambda>0$, we find $r = \max\{i:\lambda_i\ge \lambda\}$. Alternatively, this threshold can be replaced by $q$ where $r = \min\{i:\sum_{i'=1}^i\lambda_{i'}^2/\sum_{i'=1}^{\min\{p_1,p_2\}}\lambda_{i'}^2\ge q\}$. We use $q$ since it is standardized in the interval $[0,1]$, whereas $\lambda$ is unbounded.

A commonly used rule-of-thumb is to choose $q=0.9$ if the training size is sufficiently large.  To avoid rejecting the null hypothesis ${\rm H}_0$ in \eqref{eq:hypotheses} too frequently (i.e., experiencing high false alarms) when the data is noisy and the training size is small,  a small value of $q$ can be chosen, resulting in a smaller $r$.
Following the definitions in \eqref{formula:features coordinates}, we obtain 
$$
\widehat M_0 = \sum_{i=1}^{\min\{p_1,p_2\}} \hat\lambda_i  \hat u_i  \hat v_i^\top \quad \mbox{ and } \quad \hat \beta_{i,t} =  \hat u_i^\top  X_t  \hat v_i  \quad \mbox{ for } i = 1,\ldots, r.
$$ 
We also obtain $\hat\gamma_{i,t}$ for $i=1,\ldots,\min\{p_1,p_2\}$ as the singular values of $\widehat R_t = X_t- \widehat M_0$.\\

\noindent{\it Selection of $c$.} The constant $c$ in \eqref{formula:dist_free_cusum} is related to the behavior of the ARL, according to the analysis in Section \ref{sec:ARL}. It is preferable to choose a small, non-negative value of $c$ so that $\E_0[d_T]<0$ and $\E_1[d_T]>0$ hold. \cite{kim2007distribution} recommend that practitioners choose $c$ within the range of 0.01 to 0.1 based on empirical studies.

\section{Theoretical Analysis}\label{sec:analysis}

In this section, % we compute the expected difference between the in-control and out-of-control $y_t$. Then, 
we analyze the shift size in the expected values of the monitoring statistics $T_t$ from the in-control to an out-of-control state. We also discuss the behaviors of $\ARL_0$ and $\ARL_1$ of the DFILM procedure. 
For theoretical analysis, we focus on high-dimensional matrices with large dimensions $p_1, p_2$ and $p_1/p_2$ bounded away from zero. But  the algorithm works for small dimension matrices as well.

\subsection{High-dimensional asymptotic analysis}
%Mean shift in \texorpdfstring{$y_t$}{yt}}\label{subsec:meanshift-yt}
For any fixed ${i} = 1, \ldots, r$, it is straightforward to derive the in-control and out-of-control expectations of $\beta_{{i},t}$: 
$$
\E_0[\beta_{{i},t}] =  u_{i}^\top \E_0 [ X_t]  v_{i} = \lambda_{i}, \quad 
\E_1[\beta_{{i},t}] =  u_{i}^\top \E_1 [ X_t]  v_{i} = \lambda_{i} + \alpha_{i}, 
$$
where $\alpha_{i}$ is defined as the static projection of the unknown shift $A$, namely 
\begin{equation}\label{formula:alpha}
    \alpha_{i} = u_{i}^\top  A  v_{i},\quad {i} = 1, \ldots, r.
\end{equation}
\begin{theorem}[Mean shift in $\beta_{{i},t}$]
\label{thm:mean difference of nu}
For any time $t$ and index ${i} = 1,\ldots, r$, 
the difference between the in-control and out-of-control expectations of $\beta_{{i},t}$ is 
\begin{equation*}
    \E_1[\beta_{{i},t}] - \E_0[\beta_{{i},t}] =  \alpha_{i}, 
\end{equation*}
where $\alpha_{i}$ is defined in \eqref{formula:alpha}.  
\end{theorem}
From Theorem \ref{thm:mean difference of nu}, the statistic $\beta_{{i},t}$ can capture the change, provided that the differences $\alpha_{i}$ are non-zero for some $i$. By the definition of $\alpha_{i}$, when $A$ do not lie in the null space of $M_0$. %, this  we can consider $\alpha_{i}$ as an approximation of the singular values of $A$. 
In such cases, constructing statistics for detecting a shift in the expectation of $\beta_{{i},t}$ can be effective. On the other hand, if $A$ lies near the null space of $M_0$, then $\alpha_{i}$ oscillates around zero. Due to this consideration, to enhance the robustness of the detection procedure by incorporating another statistic, $\gamma_{{i},t}$. 

To analyze the in-control and out-of-control properties of the statistic $\gamma_{{i},t}$, We need the following two lemmas. 
The following lemma is a direct corollary of Theorem 2 in \cite{bai2008limit}:
\begin{lemma}
\label{lem:in-control gamma}
For any time $t<k$, suppose $R_t\in\mathbb{R}^{p_1\times p_2}$ has independent and identically distributed (i.i.d.) entries with mean zero and variance $\sigma^2$, and their fourth moments are finite. 
As $p_1,p_2\to\infty,p_1/p_2\to\eta\in(0,1)$, we have: 
    \begin{equation*}
        \sigma \left(1-\sqrt{\eta}\right) \leq \lim_{p_1,p_2\to\infty} \frac{\gamma_{{i},t}}{\sqrt{p_2}} \leq \sigma \left(1+\sqrt{\eta}\right),\quad  {i}= 1,\ldots,r, \text{ a.s.},
    \end{equation*}
where $\gamma_{{i},t}$ is the ${i}$th largest singular value of $R_t$.
\end{lemma}
The lemma provides asymptotic upper and lower bounds for the in-control $\gamma_{{i},t}$. Specifically, the order of $\gamma_{{i},t}$ is $\operatorname{\Theta}\left(\sqrt{p_2}\right)$.   In the out-of-control case, the residual matrix can be decomposed into the shift matrix $A$ and the zero-mean matrix $R_t-A$. Lemma \ref{lem:out-of-control gamma} indicates that the asymptotic behavior of $R_t$ is dominated by the deterministic limiting behavior of $ A$, rather than the random matrix $R_t-A$. The following lemma is a consequence of Theorem 1.1 in \cite{bryc2020singular}:
\begin{lemma} \label{lem:out-of-control gamma}
For $t\geq k$, suppose that $R_t-A$ have i.i.d.\@ entries with zero mean and variance $\sigma^2$, their fourth moments are finite, and the limit 
$
\lim_{p_1, p_2\to\infty}p_1/p_2 = \eta >0
$ 
exists. For ${i} = 1,\ldots, r$, let $\rho_{i}$ denote the ${i}$th largest singular value of $A$. Furthermore, assume $ \lim_{p_1,p_2\to\infty}\rho_{i}/\sqrt{p_1p_2} = \bar{\rho}_{i}$ exists and  is distinct and strictly positive, satisfying $\bar{\rho}_1>\ldots>\bar{\rho}_r>0$. Then $\gamma_{{i},t}$ can be decomposed as 
\begin{equation}\label{formula:decomposition of out-of-control gamma}
    \gamma_{{i},t} = \rho_{i} + z_{i} + m_{i} + \varepsilon_{i}, \quad {i}= 1,\ldots,r
\end{equation}
where $z_{i}$ is a random variable dependent on the dimensions $p_1$ and $p_2$ with zero mean and bounded variance, $\varepsilon_{i}$ is a random variable converging to zero in probability with respect to the dimensions $p_1$ and $p_2$, and the deterministic term $m_{i}$ satisfies 
\begin{equation*}
    m_{i} = \frac{\sigma^2}{2}\left(\frac{\sqrt{\eta}}{\bar{\rho}_{i}^3p_1p_2} - \frac{1}{\sqrt{\eta}\bar{\rho}_{i}}\right). 
\end{equation*}
\end{lemma}
This lemma decomposes the out-of-control $\gamma_{{i},t}$ defined in \eqref{formula:decomposition of out-of-control gamma} into several terms. The random noise $\varepsilon_{i}$ vanishes as the dimensions grow, and $z_{i}$ has zero mean. The deterministic term $m_{i}$ diminishes with the dimensions $p_1$ and $p_2$, converging to a constant of $-{\sigma^2}/{(2\sqrt{\eta}\bar\rho_{i})}$. Given that the singular value of the underlying anomaly $\rho_{i}$ is of the order $\operatorname{\Theta}(\sqrt{p_1p_2})$, it dominates the asymptotic behavior of $\gamma_{{i},t}$.

\begin{theorem}[Asymptotic mean shift in $\gamma_{{i},t}$]
\label{thm:mean difference of gamma}
With the assumptions in Lemmas~\ref{lem:in-control gamma} and \ref{lem:out-of-control gamma}, the limit of the expected difference between in-control and out-of-control $\gamma_{{i},t}$ is expressed as follows:
\begin{equation*}
    \lim_{p_1,p_2\to\infty} \frac{\E_1 [\gamma_{{i},t}]-\E_0 [\gamma_{{i},t}]}{\sqrt{p_1 p_2}} = \bar{\rho}_{i}.
\end{equation*}    
\end{theorem}
The proof of the theorem is available in Appendix~\ref{app:proof}.
Asymptotically, the expectation of the singular values of $R_t$ under an out-of-control state is dominated by the singular values of the shift $A$. Thus, incorporating $\gamma_t$ into statistics enhances the robustness of the detection capability of the DFLIM procedure against the algebraic relationship between $M_0$ and $A$. 

As mentioned in Theorem~\ref{thm:mean difference of nu}, if the left or right singular space of the shift $A$ is orthogonal to $\Span(u_1,\ldots,u_r)$ or $\Span(v_1,\ldots,v_r)$, $\beta_{{i},t}$ may not accurately reflect the shift, and $\gamma_{i,t}$ may help with detection. However, we choose not to rely solely on $\gamma_{i,t}$ for detection because $\beta_t$ demonstrates stronger detection power when the shift $A$ is close in shape to the in-control mean $M_0$. Consider an extreme scenario where, for some non-zero constant $c'$, the equation $A=c'M_0$ holds. For large dimensions $p_1$ and $p_2$, Theorem \ref{thm:mean difference of nu}, Lemma \ref{lem:in-control gamma}, and Lemma \ref{lem:out-of-control gamma} imply
 \begin{align}
     &\left(\E_1[\beta_{{i},t}] - \E_0[\beta_{{i},t}]\right) - \left(\E_1 [\gamma_{{i},t}]-\E_0 [\gamma_{{i},t}]\right) \nonumber\\ 
     &= \alpha_{i} - \left(\rho_{i}+\frac{\sigma^2}{2}\left(\frac{\sqrt{\eta}}{\bar{\rho}_{i}^3p_1p_2} - \frac{1}{\sqrt{\eta}\bar{\rho}_{i}}\right)+\E_1 [\varepsilon_{i}] -\E_0 [\gamma_{{i},t}]\right) \nonumber\\
     &= c'\lambda_{i} - \left(c'\lambda_{i}+\frac{\sigma^2}{2}\left(\frac{\sqrt{\eta}}{\bar{\rho}_{i}^3p_1p_2} - \frac{1}{\sqrt{\eta}\bar{\rho}_{i}}\right)+\E_1 [\varepsilon_{i}] -\E_0 [\gamma_{{i},t}]\right) \nonumber\\
     &=\frac{\sigma^2}{2}\left(\frac{1}{\sqrt{\eta}\bar{\rho}_{i}}-\frac{\sqrt{\eta}}{\bar{\rho}_{i}^3p_1p_2} \right)-\E_1 [\varepsilon_{i}] + \E_0 [\gamma_{{i},t}] \nonumber\\
     &\geq \sigma (1-\sqrt{\eta})\sqrt{p_2}+\operatorname{o}(\sqrt{p_2}) + \frac{\sigma^2}{2}\left(\frac{1}{\sqrt{\eta}\bar{\rho}_{i}}-\frac{\sqrt{\eta}}{\bar{\rho}_{i}^3p_1p_2} \right) - \E_1 [\varepsilon_{i}]. \label{formula:benefit of nu}
 \end{align}
If we impose the assumption that for any dimensions $p_1$ and $p_2$, the random variable $\varepsilon_{i}$ is uniformly integrable, then the convergence of $\varepsilon_{i}$ to zero in probability implies the convergence of the expectation of $\varepsilon_{i}$, namely $\E_1 [\varepsilon_{i}] = \operatorname{o}(1)$. On the right-hand side of inequality \eqref{formula:benefit of nu}, if $\eta > 0$, the dominant term is $\sigma (1-\sqrt{\eta})\sqrt{p_2}$, which grows to infinity with respect to dimensions $p_1$ and $p_2$. 
In a special case where $\eta$ is close to $1$, the term ${\sigma^2}\left({1}/{\sqrt{\eta}\bar{\rho}_{i}}-{\sqrt{\eta}}/{(\bar{\rho}_{i}^3p_1p_2)} \right)/2$ remains positive and still shows the advantage of the feature vector $\beta_{{i},t}$ compared with $\gamma_{i,t}$. Therefore, both statistics in $y_t$ are necessary for effective and robust detection.

\subsection{Non-asymptotic analysis allowing spatial dependence}
%Beyond spatial independence: mean shift in \texorpdfstring{$T_t$}{T2t}}

We next present an alternative non-asymptotic derivation for the expected increment term of the CUSUM procedure \eqref{formula:T2-like series}, before and after the change. This approach does not rely on lemmas derived from the asymptotic behavior of random matrices; instead, it utilizes second-order information from the detection statistics involving $y_t$.

Although the connection between the spatial dependence in $X_t$ and the structure of $\cov_{\bullet}(y_t),\, \bullet \in \{0,1\}$ is challenging to characterize-- the covariances $\cov_\bullet(\gamma_{t})$ and $\cov_\bullet(\beta_{t}, \gamma_{t})$ remain unknown-- because $\gamma_{t}$ consists of the singular values of a (possibly non-central) random matrix with correlated entries. Instead of delving into the structure of $\cov_\bullet(y_t)$ based on specific spatial dependencies, we express the covariances of $\beta_t$ and $\gamma_t$ in a sufficiently general form to capture all potential structures arising from the spatial dependence in $X_t$. This approach allows the analysis of $T_t$ to accommodate cases where spatial dependence exists in $X_t$.

The in-control and out-of-control covariance matrices of $y_t$ as follows:
\begin{equation*}
    \cov_0(y_t) = \Sigma = \begin{bmatrix}
        \Sigma_{\beta} & P \\
        P^\top & \Sigma_{\gamma} \\
    \end{bmatrix}, \quad 
    \cov_1(y_t) = \tilSig = \begin{bmatrix}
        \Sigma_{\beta} & \tilP \\
        \tilP^\top & \tilSig_{\gamma} \\
    \end{bmatrix}, 
\end{equation*}
where 
$$
\Sigma_\beta = \cov_0\cirbrack{\beta_t},\  P = \cov_0\cirbrack{\beta_t,\gamma_t},\  \tilP = \cov_1\cirbrack{\beta_t,\gamma_t},\  \Sigma_\gamma = \cov_0\cirbrack{\gamma_t},\ \mbox{ and } \tilSig_\gamma = \cov_1\cirbrack{\gamma_t}.
$$
Clearly, $\Sigma_\beta = \cov_1\cirbrack{\beta_t}$ when only a mean shift is assumed, as in a typical mean-shift detection problem. Corresponding to the block structures that differentiate $\beta$ and $\gamma$, we also define the expected shift size in them as $\E_1\sqbrack{y_t} - \E_0\sqbrack{y_t}=\delta = [\delta_\beta^\top,\delta_\gamma^\top]^\top$.
\begin{theorem}\label{lem:beyond iid gain in T square}
The difference between the in-control and out-of-control expectations of the $T_t$ statistics is expressed as follows: 
    \begin{equation*}
    \begin{aligned}
        \E_1\sqbrack{T_t} - \E_0\sqbrack{T_t} =\Delta_1+\Delta_2+\Delta_3,
    \end{aligned}
    \end{equation*}
    where 
    \begin{align*}
        \Delta_1 &= \tr\left[\left(\Sigma/\Sigma_\beta\right)^{-1}
       \left(\tilSig/\Sigma_{\beta}\right)\right]-r,\\[2pt]
       \Delta_2 &= \tr\left\{\left(\Sigma/\Sigma_\beta\right)^{-1}\left[
        \Sigma_{\beta}^{-1/2}\left(P-\tilP\right)\right]^\top\left[
        \Sigma_{\beta}^{-1/2}\left(P-\tilP\right)\right]\right\},\\[2pt]
        \Delta_3 &= \left\|\Sigma_\beta^{-1/2}\delta_\beta\right\|^2 + \left\|\left(\Sigma/\Sigma_\beta\right)^{-1/2}\left(P^\top\Sigma_{\beta}^{-1}\delta_\beta+\delta_\gamma\right)\right\|^2, \\[2pt]
        \Sigma/&\Sigma_\beta := \Sigma_\gamma-P^\top \Sigma_\beta^{-1} P, \quad \mbox{ and } \quad \tilSig/\Sigma_\beta := \tilSig_\gamma-\tilP^\top \Sigma_\beta^{-1} \tilP.
    \end{align*}
\end{theorem}
The proof of the theorem is provided in Appendix~\ref{app:proof}. According to Theorem~\ref{lem:beyond iid gain in T square}, $\Delta_1$ is small when $\tilSig$ deviates slightly from $\Sigma$. In contrast, $\Delta_2$ is positive since it is the trace of the multiplication of two positive definite matrices. The last term, $\Delta_3$, increases quadratically with the shift sizes in $\beta_t$ and $\gamma_t$. Assuming non-zero shift sizes in $\beta_t$ or $\gamma_t$, $\Delta_3$ is also positive as validated by Theorems~\ref{thm:mean difference of nu} and~\ref{thm:mean difference of gamma} under the assumption of independence among matrix entries. Overall, Theorem~\ref{lem:beyond iid gain in T square} demonstrates the approximate positivity of the mean shift in $T_t$, enabling DFLIM to raise an alarm after the change point quickly.

\subsection{\texorpdfstring{$\ARL_0$}{ARL0} and \texorpdfstring{$\ARL_1$}{ARL1}}\label{sec:ARL}
In this section, we provide expressions of $\ARL_0$ and $\ARL_1$. Recall that $\ARL_0$ and $\ARL_1$ represent $\E_0[\tau]$ and $\E_1[\tau]$, respectively. According to Lemma~\ref{lem:ARL}, for the in-control case, we  have:
\begin{equation*}
\ARL_0 = \E_0\left[\tau\right] \approx 
\frac{\Omega_0^2}{2(c\sigma_T)^2}\left[\exp \left(\frac{2c\sigma_TH}{\Omega_0^2}\right)-1-\frac{2c\sigma_TH}{\Omega_0^2}\right],
\end{equation*}
where $\ARL_0$ increases exponentially with respect to the control limit $H$. For an out-of-control case, we obtain:
\begin{equation*}
\ARL_1 = \E_1\left[\tau\right] \approx 
\frac{\Omega_1^2}{2 d_1^2}\left[\exp \left(-\frac{2Hd_1}{\Omega_1^2}\right)-1+\frac{2d_1 H}{\Omega_1^2}\right],
\end{equation*}
where $\Omega_1^2$ represents $\Omega^2$ under the out-of-control probability measure and $d_1 = \E_1[T_t] - \E_0[T_t] -c\sigma_T$. 
Given the approximate positivity of $d_1$ as validated in Theorem~\ref{lem:beyond iid gain in T square}, $\ARL_1$ increases on the order of $\OO(H)$ (i.e., increases linearly with respect to $H$). Consequently, for a large $H$, we have $\ARL_1 \ll \ARL_0$. 

To study the behavior of $\ARL_1$ when the assumption of unchanged $\Omega^2$ is violated, we consider the case where the process variability shifts from $\Omega_0^2$ to $\Omega_1^2$. Performing a Taylor expansion on $\E_1[\tau]$, we obtain
\[
\begin{aligned}
  \E_1[\tau] &\approx \frac{\Omega_1^2}{2 d_1^2}\left[1-\frac{2 d_1 H}{\Omega_1^2}+\left(\frac{2 d_1 H}{\Omega_1^2}\right)^2+{o}\left(\frac{1}{\Omega_1^4}\right)-1+\frac{2 d_1 H}{\Omega_1^2}\right]\\
    &= \frac{2H^2}{\Omega_1^2}+   o\left(\frac{1}{\Omega_1^2}\right).
\end{aligned}
\]
This approximation illustrates that if $\Omega_1^2 > \Omega_0^2$, the detection of a shift becomes even faster compared to the case of unchanged $\Omega^2$. On the other hand, $\Omega_1^2$ is smaller than $\Omega_0^2$, detection becomes slower. However, even with a decrease in $\Omega_1^2$, the reduction is typically not substantial compared to the size of $\Omega_0^2$, which tends to be large in practice. Therefore, the slowdown in detection is usually not significant. 

\section{Simulated Experiments}\label{sec:simulation}

In this section, we perform numerical experiments to compare the DFLIM procedure with existing procedures. Specifically, we consider the following three baselines: (1) MEWMA \citep{otto2021parallelized}, which applies parallelized multivariate EWMA control charts to matrix data for efficient detection of local changes; (2) MGLR \citep{okhrin2021new}, which constructs monitoring statistics based on ROI and corresponding likelihood ratios; (3) ST-SSD \citep{yan2018real}, which decomposes the data into informative characteristics and noises. 

\subsection{Settings}
\label{subsec:settings}

First, we generate in-control matrices $\{X_t\in\mathbb R^{p_1\times p_2} : t = 1,2,\ldots\}$ using the following equation:
\begin{equation}\label{formula:generation}
X_t = M_0 + \sum_{j = 0}^{\ell}\phi^j \varepsilon_{t-j},
\end{equation}
where $M_0$ is the low-rank in-control mean matrix, $\varepsilon_t$ denote noise matrices, and $\phi>0$ is the auto-correlation parameter. Equation \eqref{formula:generation} represents a moving average model of order $\ell$. In our simulated experiments, we set $p_1=100$ and $p_2=200$. The in-control data are generated as follows:
\begin{itemize}
    \item We vary the ranks of $M_0$ within the set $\{2,5\}$. The rank-two $M_0$ is constructed as a chessboard with elements $M_0[j_1,j_2]$ for $j_1=1,2,\ldots, 100$ and $j_2=1,2,\ldots,200$:
$$
M_0[j_1,j_2] =
\begin{cases}
0.1, &10k_1+1\leq j_1 \leq 10k_1+5,\ 40k_2 + 11 \leq  j_2 \leq 40k_2 + 20, \\
0.1, &10k_1+6\leq j_1 \leq 10k_1+10,\ 40k_2 + 21 \leq  j_2 \leq 40k_2 + 30, \\
-0.1, &10k_1+1\leq j_1 \leq 10k_1+5,\ 40k_2 + 31 \leq j_2 \leq 40k_2 + 40, \\
-0.1, &10k_1+6\leq j_1 \leq 10k_1+10,\ 40k_2 + 1 \leq j_2 \leq 40k_2 + 10, \\
0, &\text{ otherwise,}
\end{cases}
$$
where $0\leq k_1 \leq 4$ and $0\leq k_2\leq 9$. To extend the rank to five, we superimpose the rank-three approximation of the truncated first image of solar flare image data \citep{xie2012change} onto the chessboard signal.  Figure \ref{fig:background} displays the resulting images of the rank-two and rank-five $M_0$.

\begin{figure}[b!]
\centering
\subfigure[$\Rank(M_0)=2$]{%
\resizebox*{5cm}{!}{\includegraphics{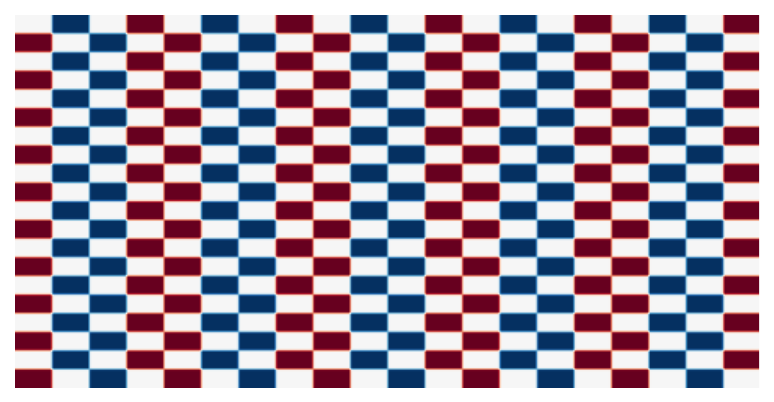}}}\hspace{25pt}
\subfigure[$\Rank(M_0)=5$]{%
\resizebox*{5cm}{!}{\includegraphics{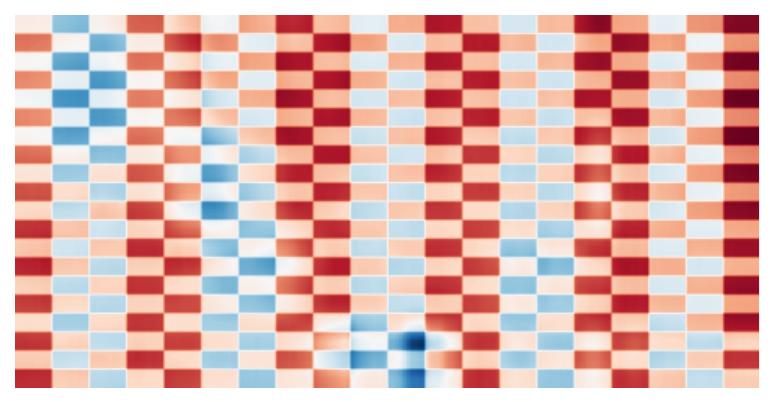}}}
\caption{In-control mean matrix $M_0$ with rank $r \in \{2,5\}$.} 
\label{fig:background}
\end{figure}
% \begin{figure}[b!]
% \centering 
% \begin{subfigure}[h]{0.3\linewidth}
% \includegraphics[width=\linewidth]{plots/Mrank2.pdf}
% \caption{$\Rank(M_0)=2$}
% \end{subfigure}
% \hspace{+15pt}
% \begin{subfigure}[h]{0.3\linewidth}
% \includegraphics[width=\linewidth]{plots/Mrank5.pdf}
% \caption{$\Rank(M_0)=5$}
% \end{subfigure}
% \caption{In-control mean matrix $M_0$ with the rank $r \in \{2,5\}$. }
% \label{fig:background}
% \end{figure}
\item For the distribution of $\varepsilon_t$, we select one from the following two distributions:
\begin{enumerate}[(1)]
    \item \textit{Normal distribution:} We generate $\varepsilon_t$ using the matrix normal distribution \citep[Chapter~2]{gupta2018matrix}. Thus,
    \begin{equation}\label{formula:noise matrix generation 1}
        \varepsilon_t \overset{\mathrm{i.i.d.}}{\sim} \mathcal{MN} \left(0,\Sigma_{\rm row}, \Sigma_{\rm col}\right),   
    \end{equation}
    where $\Sigma_{\rm row}\in\mathbb R^{p_1\times p_1}$ and $\Sigma_{\rm col}\in\mathbb R^{p_2\times p_2}$ are specified to capture the spatial correlations.  
    \item \textit{Non-normal distribution:} We first generate $\tilde \varepsilon_t$ from \eqref{formula:noise matrix generation 1} and then transform each entry $\tilde \varepsilon_t[j_1,j_2]$ into $\varepsilon_t[j_1,j_2]$ as follows: 
    \begin{equation*}
        \varepsilon_t[j_1,j_2] = -\log(1 - \Psi(\tilde\varepsilon_t[j_1,j_2])),
    \end{equation*}
    where $\Psi(\cdot)$ represents the cumulative distribution of the standard normal random variable. Thus, $\varepsilon_t$ has entries $\varepsilon_t[j_1,j_2]$ that are exponentially distributed with an expectation of $1$ and exhibit correlations among them. 
\end{enumerate}
\item To incorporate spatial correlation, both covariance matrices $\Sigma_{\rm row}$ and $\Sigma_{\rm col}$ are specified using either tri-diagonal covariance or exponential covariance, defined as follows:
\begin{equation}\label{formula:spatial correlation}
\begin{aligned}
\text{Tri-diagonal } &= \begin{bmatrix}
1 & 0.3 & \cdots & 0 & 0 \\
0.3 & 1 & \cdots & 0 & 0 \\
\vdots & \vdots & \ddots   & \vdots & \vdots \\
0 & 0 & \cdots & 1 & 0.3 \\
0 & 0 & \cdots & 0.3 & 1
\end{bmatrix},\\
\text{Exponential } &= 
\begin{bmatrix}
1 & 0.3 &  \cdots & 0.3^{p-2} & 0.3^{p-1} \\
0.3 & 1 &  \cdots & 0.3^{p-3} & 0.3^{p-2} \\
\vdots  & \vdots & \ddots & \vdots & \vdots \\
0.3^{p-2} & 0.3^{p-3} & \cdots  & 1 & 0.3 \\
0.3^{p-1} & 0.3^{p-2} & \cdots  & 0.3 & 1
\end{bmatrix},
\end{aligned}
\end{equation}
where the dimensions of the covariance matrices vary in $\{p_1, p_2\}$, corresponding to the row and column sides. 
\item To incorporate auto-correlation, we vary the lag order $\ell \in \{5, 20\}$ 
for the moving average term $\sum_{j = 0}^{\ell}\phi^j \varepsilon_{t-j}$ in \eqref{formula:generation}, with the parameter $\phi$ fixed as $\phi = 0.5$. 
\end{itemize}

We generate the out-of-control matrices according to the following equation:
\begin{equation*}\label{formula:OOC generation}
X_t = M_0 + A + \sum_{j = 0}^{\ell}\phi^j \varepsilon_{t-j},
\end{equation*}
where the mean $M_0$ and the noise $\sum_{i = 0}^{\ell}\phi^i \varepsilon_{t-i}$ follow the same settings as in the in-control matrices. 
We test the following four mean shifts for $A$:
\begin{enumerate}[(i)]
    \item Sparse: 
     $$ A[j_1,j_2] = 
    \begin{cases}
    3, &\text{ if } 8\leq j_1\leq 13 \text{ and } 18\leq j_2\leq 23, \\
    0, &\text{ otherwise.}
    \end{cases}
    $$
    \item Ring: 
    $$
    A[j_1,j_2] =
    \begin{cases}
    0.173,& \text{ if } \floor{\sqrt{(j_1-50)^2+(j_2-100)^2}} = 12k_1 + k_2, \\
    -0.173,& \text{ if } \floor{\sqrt{(j_1-50)^2+(j_2-100)^2}} = 12k_1 + 8 + k_2, \\
    0, &\text{ otherwise,}
    \end{cases}
    $$
    where $k_1$ is some non-negative integer and $0 \leq k_2 \leq 3$.  
    \item Sine: 
    $$
    A[j_1,j_2] = 
    0.283\sin{\frac{j_2\pi}{5}}\sin{\frac{2j_1\pi}{5}}.
    $$
    \item Chessboard: identical to the rank-two $M_0$. 
\end{enumerate}
The visual representation of the four shift patterns is provided in Figure \ref{fig:shifts}.
\begin{figure}[t]
\centering
\subfigure[Sparse]{\label{fig:sparse}
\resizebox*{0.22\linewidth}{!}{\includegraphics{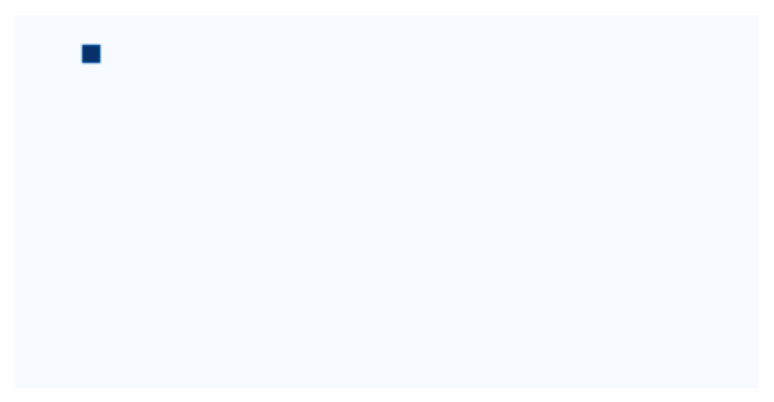}}}\hspace{0pt}
\subfigure[Ring]{\label{fig:ring}
\resizebox*{0.22\linewidth}{!}{\includegraphics{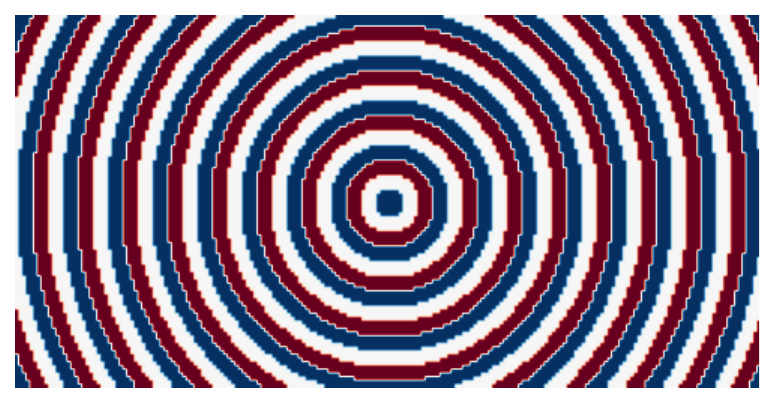}}}\hspace{0pt}
\subfigure[Sine]{\label{fig:sine}
\resizebox*{0.22\linewidth}{!}{\includegraphics{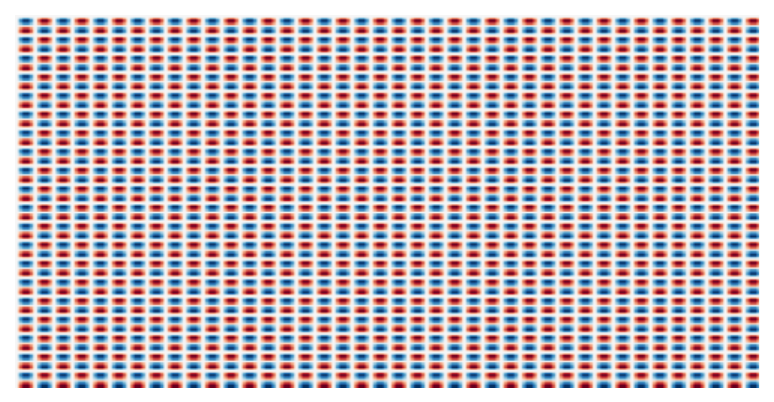}}}\hspace{0pt}
\subfigure[Chessboard]{\label{fig:chessboard}
\resizebox*{0.22\linewidth}{!}{\includegraphics{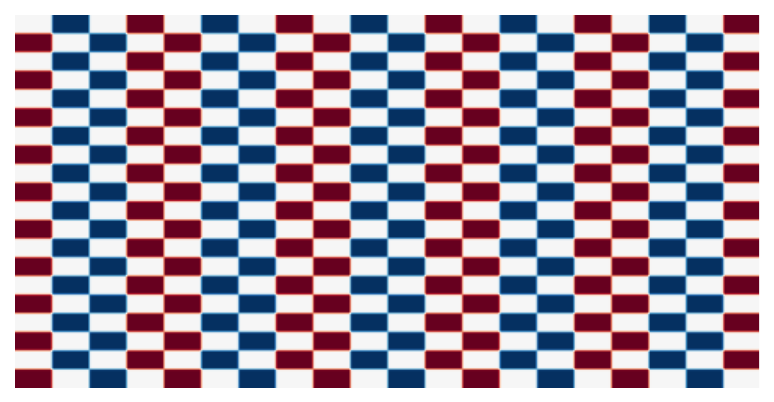}}}\hspace{0pt}
\caption{Four patterns of shift matrix $A$.}
\label{fig:shifts}
\end{figure}
% \begin{figure}[t]
% \centering 
% \begin{subfigure}[h]{0.22\linewidth}
% \includegraphics[width=\linewidth]{plots/Sparse.pdf}
% \caption{Sparse}
% \label{fig:sparse}
% \end{subfigure}
% \begin{subfigure}[h]{0.22\linewidth}
% \includegraphics[width=\linewidth]{plots/Ring.pdf}
% \caption{Ring}
% \label{fig:ring}
% \end{subfigure}
% \begin{subfigure}[h]{0.22\linewidth}
% \includegraphics[width=\linewidth]{plots/Sine.pdf}
% \caption{Sine}
% \label{fig:sine}
% \end{subfigure}
% \begin{subfigure}[h]{0.22\linewidth}
% \includegraphics[width=\linewidth]{plots/Chessboard.pdf}
% \caption{Chessboard}
% \label{fig:chessboard}
% \end{subfigure}
% \caption{Four patterns of the shift matrix $A$.}
% \label{fig:shifts}
% \end{figure}

\subsection{In-control performance}
In each in-control setting, we generate 1000 sequences, each consisting of $800$ matrix observations. The target $\ARL_0$ is set to $200$. Baselines MEWMA, MGLR, and ST-SSD tune their control limits through trial-and-error using these 1000 sequences. For DFLIM, we set $c=0.01$ for all experiments and determine the control limit by solving equation \eqref{formula:control limit} for $H$. Each procedure is then applied to these 1000 sequences using its determined control limit to estimate $\ARL_0$.

In each in-control setting, we generate 1000 sequences, each consisting of 800 matrix observations. The target $\ARL_0$ is set to 200. Baseline methods MEWMA, MGLR, and ST-SSD tune their control limits through trial-and-error using these 1000 sequences. For DFLIM, we determine the control limit by solving equation \eqref{formula:control limit} for $H$. Subsequently, each procedure is applied to the same dataset to find the control limits.

Table~\ref{tab:arl0} summarizes the control limits ($H$) and estimated $\ARL_0$ values for different procedures. 
 MEWMA, MGLR, and ST-SSD achieve $\ARL_0$ values close to the target 200, as expected, because their control limits are determined through trial-and-error calibration using the same data sequences. The table also demonstrates that the analytically determined control limits for DFLIM provide relatively accurate $\ARL_0$ values close to the target. The simulated data has complicated spatial correlations as defined in \eqref{formula:spatial correlation}, and the target $\ARL_0$ is set to a small value, potentially increasing the risk that the time horizon might not be long enough for the FCLT to be applicable. Despite these slight theoretical violations, the empirical $\ARL_0$ values obtained for DFLIM are still quite accurate. These results suggest that \eqref{formula:control limit} is widely applicable in practice.

\begin{table}[t]
\caption{Control limit ($H$) and estimated $\ARL_0$ for various settings of simulated processes with target $\ARL_0=200$ (standard errors in parentheses). }
\label{tab:arl0}
\resizebox{\textwidth}{!}{
% \begin{scriptsize}
\begin{tabular}{cccc cccc cccc} 
\toprule
  &  &  &  & \multicolumn{2}{c}{MEWMA} & \multicolumn{2}{c}{MGLR} & \multicolumn{2}{c}{ST-SSD}  & \multicolumn{2}{c}{DFLIM} \\
  \cmidrule(lr){1-4}   \cmidrule(lr){5-6}   \cmidrule(lr){7-8}   \cmidrule(lr){9-10}   \cmidrule(lr){11-12}
Distribution & Rank & Lag & Covariance  & {$H$} & {$\ARL_0$} & {$H$} & {$\ARL_0$} & {$H$} & {$\ARL_0$} & {$H$} & {$\ARL_0$}\\ 
  \cmidrule(lr){1-4}   \cmidrule(lr){5-6}   \cmidrule(lr){7-8}   \cmidrule(lr){9-10}   \cmidrule(lr){11-12}
Normal & 2 & 5  & Tri-diagonal & 180.202 & 198.62 (6.248) & 6.062 & 201.04 (5.495) & 2.800 & 201.05 (5.680) & 36.507 & 201.48 (5.321) \\
Normal & 2 & 5  & Exponential & 180.219 & 200.41 (5.910) & 6.043 & 200.72 (5.528) & 2.799 & 199.42 (5.692) & 36.654 & 197.26 (5.119) \\
Normal & 2 & 20 & Tri-diagonal & 180.455 & 200.20 (5.866) & 6.155 & 198.46 (5.465) & 2.796 & 199.87 (6.036) & 36.776 & 203.02 (5.188) \\
Normal & 2 & 20 & Exponential & 180.415 & 198.12 (5.890) & 6.236 & 200.17 (5.574) & 2.804 & 200.17 (5.698) & 36.935 & 200.81 (5.103) \\
Normal & 5 & 5  & Tri-diagonal & 184.765 & 198.72 (8.941) & 6.039 & 200.94 (5.694) & 2.779 & 198.83 (5.576) & 36.551 & 201.30  (5.297) \\
Normal & 5 & 5  & Exponential & 183.968 & 201.29 (9.059) & 6.015 & 198.83 (5.469) & 2.796 & 201.10 (5.715) & 36.688 & 201.28 (5.173) \\
Normal & 5 & 20 & Tri-diagonal & 184.236 & 201.88 (9.081) & 6.221 & 199.75 (5.640) & 2.808 & 200.98 (5.687) & 36.632 & 205.51 (5.354) \\
Normal & 5 & 20 & Exponential & 184.509 & 198.50 (9.053) & 6.203 & 200.92 (5.574) & 2.783 & 198.36 (5.827) & 36.641 & 197.76 (5.167) \\
Non-normal & 2  & 5 & Tri-diagonal & 294.912 & 199.32 (10.903) & 6.042 & 201.49 (5.562) & 2.837 & 198.39 (5.935) & 37.208 & 202.81 (5.167) \\
Non-normal & 2  & 5 & Exponential & 296.649 & 201.74 (10.953) & 6.013 & 200.26 (5.526) & 2.814 & 200.38 (5.988) & 37.416 & 200.32 (5.292) \\
Non-normal & 2  & 20 & Tri-diagonal & 297.494 & 199.15 (10.911) & 6.191 & 201.17 (5.839) & 2.828 & 198.72 (5.883) & 37.359 & 204.95 (5.187) \\
Non-normal & 2  & 20 & Exponential & 299.497 & 200.49 (10.930) & 6.173 & 198.74 (5.435) & 2.834 & 200.51 (5.846) & 37.510 & 202.34 (5.425) \\
Non-normal & 5  & 5 & Tri-diagonal & 389.186 & 199.95 (10.926) & 6.017 & 199.69 (5.446) & 2.821 & 199.74 (5.596) & 37.369 & 205.32 (5.383) \\
Non-normal & 5  & 5 & Exponential & 391.227 & 199.95 (10.926) & 5.988 & 200.12 (5.441) & 2.832 & 199.70 (5.634) & 37.283 & 206.09 (5.357) \\
Non-normal & 5  & 20 & Tri-diagonal & 394.891 & 199.15 (10.911) & 6.212 & 200.55 (5.578) & 2.818 & 200.67 (5.817) & 37.483 & 207.76 (5.351) \\
Non-normal & 5  & 20 & Exponential & 394.642 & 199.95 (10.926) & 6.215 & 201.63 (5.533) & 2.834 & 200.08 (5.867) & 37.565 & 206.17 (5.253) \\
\bottomrule
\end{tabular}}
% \end{scriptsize}
\end{table}

\subsection{Out-of-control performance}
In out-of-control experiments, we consider shift matrix $A$ with sparse, ring, sine, and chessboard patterns as described in Section~\ref{subsec:settings}, which are added to the in-control data. The experiment is repeated 1000 times for each procedure. The results of the out-of-control performances are summarized in Tables \ref{tab:sparse arl1} --\ref{tab:chessboard arl1} for each shift pattern $A$. 

\begin{table}[tbp]
\caption{Sparse shift: $\ARL_1$ for various settings of simulated processes with target $\ARL_0=200$ (standard errors in parentheses). }
\label{tab:sparse arl1}
\centering
\resizebox{\textwidth}{!}{
\begin{tabular}{cccc  cccc} \toprule
{Distribution} & {Rank} & {Lag} & {Covariance} 
 & \multicolumn{1}{c}{MEWMA} & \multicolumn{1}{c}{MGLR} & \multicolumn{1}{c}{ST-SSD}  & \multicolumn{1}{c}{DFLIM}\\ 
  \cmidrule(lr){1-4}  \cmidrule(lr){5-8}   
Normal & 2 & 5 & Tri-diagonal & 198.93 (6.227) & 207.34 (5.756) & 79.17 (2.347) & 15.06 (0.232) \\
Normal & 2 & 5 & Exponential & 197.69 (6.249) & 200.58 (5.425) & 85.01 (2.603) & 16.66 (0.289) \\
Normal & 2 & 20 & Tri-diagonal & 196.08 (6.190) & 188.46 (5.211) & 80.64 (2.597) & 15.61 (0.264) \\
Normal & 2 & 20 & Exponential & 201.89 (6.215) & 200.11 (5.618) & 84.05 (2.600) & 16.49 (0.278) \\
Normal & 5 & 5 & Tri-diagonal & 180.50 (9.163) & 202.62 (5.683) & 74.58 (2.334) & 14.77 (0.233) \\
Normal & 5 & 5 & Exponential & 155.01 (8.373) & 197.18 (5.409) & 84.31 (2.661) & 16.09 (0.271) \\
Normal & 5 & 20 & Tri-diagonal & 153.96 (8.399) & 198.62 (5.486) & 83.17 (2.589) & 15.18 (0.258) \\
Normal & 5 & 20 & Exponential & 175.24 (8.724) & 202.35 (5.428) & 77.40 (2.417) & 16.39 (0.275) \\
Non-normal & 2 & 5 & Tri-diagonal & 156.04 (9.991) & 197.86 (5.533) & 83.06 (2.675) & 20.52 (0.353) \\
Non-normal & 2 & 5 & Exponential & 178.38 (10.501) & 194.10 (5.516) & 84.74 (2.624) & 21.11 (0.387) \\
Non-normal & 2 & 20 & Tri-diagonal & 163.70 (10.163) & 191.24 (5.293) & 82.38 (2.530) & 21.03 (0.372) \\
Non-normal & 2 & 20 & Exponential & 174.38 (10.415) & 189.54 (5.334) & 85.97 (2.708) & 21.20 (0.388) \\
Non-normal & 5 & 5 & Tri-diagonal & 164.00 (10.182) & 202.14 (5.522) & 85.50 (2.744) & 21.03 (0.372) \\
Non-normal & 5 & 5 & Exponential & 176.78 (10.467) & 182.21 (5.268) & 87.27 (2.784) & 21.78 (0.396) \\
Non-normal & 5 & 20 & Tri-diagonal & 182.37 (10.584) & 195.56 (5.465) & 80.65 (2.769) & 20.80 (0.379) \\
Non-normal & 5 & 20 & Exponential & 156.80 (10.011) & 196.41 (5.581) & 90.09 (2.824) & 22.01 (0.397) \\
\bottomrule
\end{tabular}}
\end{table}
\begin{table}[tbp]
\caption{Ring shift: $\ARL_1$ for various settings of simulated processes with target $\ARL_0=200$ (standard errors in parentheses). }
\label{tab:ring arl1}
\centering
\resizebox{\textwidth}{!}{
\begin{tabular}{cccc  cccc} \toprule
{Distribution} & {Rank} & {Lag} & {Covariance} 
 & \multicolumn{1}{c}{MEWMA} & \multicolumn{1}{c}{MGLR} & \multicolumn{1}{c}{ST-SSD}  & \multicolumn{1}{c}{DFLIM}\\ 
  \cmidrule(lr){1-4}  \cmidrule(lr){5-8}   
Normal & 2 & 5 & Tri-diagonal & 196.30 (6.062) & 194.71 (5.712) & 21.50 (0.687) & 28.69 (0.498) \\
Normal & 2 & 5 & Exponential & 197.92 (6.041) & 209.54 (5.775) & 23.73 (0.758) & 27.41 (0.476) \\
Normal & 2 & 20 & Tri-diagonal & 199.30 (6.110) & 190.31 (5.435) & 20.05 (0.610) & 29.10 (0.490) \\
Normal & 2 & 20 & Exponential & 197.52 (5.880) & 198.82 (5.691) & 24.04 (0.734) & 26.81 (0.435) \\
Normal & 5 & 5 & Tri-diagonal & 204.50 (9.416) & 198.12 (5.554) & 20.22 (0.579) & 27.86 (0.461) \\
Normal & 5 & 5 & Exponential & 159.62 (8.229) & 194.84 (5.631) & 24.05 (0.728) & 26.20 (0.440) \\
Normal & 5 & 20 & Tri-diagonal & 175.98 (8.595) & 203.60 (5.518) & 22.42 (0.661) & 28.15 (0.490) \\
Normal & 5 & 20 & Exponential & 184.53 (8.824) & 195.35 (5.491) & 24.35 (0.718) & 27.19 (0.459) \\
Non-normal & 2 & 5 & Tri-diagonal & 160.79 (10.106) & 195.68 (5.555) & 21.16 (0.648) & 47.58 (0.929) \\
Non-normal & 2 & 5 & Exponential & 193.44 (10.800) & 208.74 (5.761) & 24.00 (0.746) & 40.47 (0.726) \\
Non-normal & 2 & 20 & Tri-diagonal & 149.26 (9.814) & 200.90 (5.598) & 20.34 (0.626) & 45.90 (0.869) \\
Non-normal & 2 & 20 & Exponential & 181.07 (10.550) & 193.64 (5.278) & 23.34 (0.710) & 42.59 (0.813) \\
Non-normal & 5 & 5 & Tri-diagonal & 190.36 (10.744) & 195.06 (5.407) & 20.66 (0.633) & 45.75 (0.852) \\
Non-normal & 5 & 5 & Exponential & 192.76 (10.791) & 185.25 (5.259) & 25.85 (0.794) & 39.06 (0.709) \\
Non-normal & 5 & 20 & Tri-diagonal & 178.38 (10.501) & 200.02 (5.537) & 20.61 (0.639) & 44.73 (0.869) \\
Non-normal & 5 & 20 & Exponential & 191.16 (10.760) & 190.57 (5.134) & 26.90 (0.839) & 41.36 (0.805) \\
\bottomrule
\end{tabular}}
\end{table}
\begin{table}[tbp]
\caption{Sine shift: $\ARL_1$ for various settings of simulated processes with target $\ARL_0=200$ (standard errors in parentheses). }
\label{tab:sine arl1}
\centering
\resizebox{\textwidth}{!}{
\begin{tabular}{cccc  cccc} \toprule
{Distribution} & {Rank} & {Lag} & {Covariance} 
 & \multicolumn{1}{c}{MEWMA} & \multicolumn{1}{c}{MGLR} & \multicolumn{1}{c}{ST-SSD}  & \multicolumn{1}{c}{DFLIM}\\ 
  \cmidrule(lr){1-4}  \cmidrule(lr){5-8}   
Normal & 2 & 5 & Tri-diagonal & 201.12 (5.933) & 199.93 (5.511) & 31.12 (0.996) & 5.29 (0.081) \\
Normal & 2 & 5 & Exponential & 209.35 (6.448) & 198.88 (5.576) & 37.34 (1.169) & 16.17 (0.244) \\
Normal & 2 & 20 & Tri-diagonal & 191.17 (5.963) & 187.88 (5.256) & 32.41 (1.047) & 5.54 (0.086) \\
Normal & 2 & 20 & Exponential & 194.94 (6.011) & 203.80 (5.514) & 35.90 (1.062) & 16.45 (0.263) \\
Normal & 5 & 5 & Tri-diagonal & 210.28 (9.451) & 192.23 (5.592) & 29.65 (0.923) & 5.45 (0.083) \\
Normal & 5 & 5 & Exponential & 176.80 (8.450) & 187.72 (5.215) & 35.86 (1.076) & 16.31 (0.273) \\
Normal & 5 & 20 & Tri-diagonal & 172.29 (8.513) & 203.99 (5.598) & 30.92 (0.909) & 5.54 (0.087) \\
Normal & 5 & 20 & Exponential & 195.50 (9.001) & 192.17 (5.401) & 35.05 (1.029) & 16.5 (0.266) \\
Non-normal & 2 & 5 & Tri-diagonal & 156.80 (10.011) & 205.74 (5.920) & 35.32 (1.079) & 6.50 (0.097) \\
Non-normal & 2 & 5 & Exponential & 192.28 (10.775) & 197.62 (5.561) & 37.04 (1.116) & 16.87 (0.255) \\
Non-normal & 2 & 20 & Tri-diagonal & 165.11 (10.197) & 184.36 (5.086) & 35.08 (1.060) & 6.59 (0.098) \\
Non-normal & 2 & 20 & Exponential & 170.39 (10.327) & 186.27 (5.357) & 39.12 (1.249) & 16.68 (0.254) \\
Non-normal & 5 & 5 & Tri-diagonal & 183.17 (10.600) & 198.31 (5.558) & 33.06 (0.970) & 6.60 (0.096) \\
Non-normal & 5 & 5 & Exponential & 209.54 (11.097) & 187.94 (5.486) & 40.92 (1.226) & 17.27 (0.266) \\
Non-normal & 5 & 20 & Tri-diagonal & 218.33 (11.243) & 193.34 (5.370) & 34.40 (1.094) & 6.40 (0.100) \\
Non-normal & 5 & 20 & Exponential & 169.59 (10.309) & 203.86 (5.566) & 39.21 (1.262) & 16.41 (0.260) \\
\bottomrule
\end{tabular}}
\end{table}
\begin{table}[tbp]
\caption{Chessboard shift: $\ARL_1$ for various settings of simulated processes with target $\ARL_0=200$ (standard errors in parentheses). }
\label{tab:chessboard arl1}
\centering
\resizebox{\textwidth}{!}{
\begin{tabular}{cccc  cccc} \toprule
{Distribution} & {Rank} & {Lag} & {Covariance} 
 & \multicolumn{1}{c}{MEWMA} & \multicolumn{1}{c}{MGLR} & \multicolumn{1}{c}{ST-SSD}  & \multicolumn{1}{c}{DFLIM}\\ 
  \cmidrule(lr){1-4}  \cmidrule(lr){5-8}   
Normal & 2 & 5 & Tri-diagonal & 190.65 (6.146) & 205.80 (5.895) & 1.54 (0.030) & 1.70 (0.017) \\
Normal & 2 & 5 & Exponential & 193.27 (6.081) & 207.11 (5.744) & 1.78 (0.044) & 1.97 (0.018) \\
Normal & 2 & 20 & Tri-diagonal & 186.63 (6.145) & 184.59 (4.887) & 1.52 (0.031) & 1.69 (0.017) \\
Normal & 2 & 20 & Exponential & 199.48 (6.307) & 199.14 (5.624) & 1.74 (0.037) & 1.97 (0.018) \\
Normal & 5 & 5 & Tri-diagonal & 163.28 (8.942) & 203.00 (5.583) & 1.49 (0.030) & 2.28 (0.021) \\
Normal & 5 & 5 & Exponential & 135.63 (7.913) & 199.87 (5.557) & 1.76 (0.039) & 2.62 (0.024) \\
Normal & 5 & 20 & Tri-diagonal & 133.09 (7.833) & 192.75 (5.604) & 1.54 (0.031) & 2.28 (0.021) \\
Normal & 5 & 20 & Exponential & 146.96 (8.206) & 198.12 (5.508) & 1.77 (0.041) & 2.59 (0.025) \\
Non-normal & 2 & 5 & Tri-diagonal & 162.08 (10.129) & 207.74 (5.870) & 1.56 (0.030) & 2.47 (0.023) \\
Non-normal & 2 & 5 & Exponential & 171.71 (10.349) & 196.14 (5.447) & 1.75 (0.036) & 2.77 (0.026) \\
Non-normal & 2 & 20 & Tri-diagonal & 131.17 (9.328) & 197.06 (5.611) & 1.57 (0.031) & 2.53 (0.023) \\
Non-normal & 2 & 20 & Exponential & 157.46 (10.022) & 195.98 (5.507) & 1.91 (0.047) & 2.80 (0.028) \\
Non-normal & 5 & 5 & Tri-diagonal & 160.00 (10.088) & 195.82 (5.581) & 1.58 (0.031) & 2.52 (0.024) \\
Non-normal & 5 & 5 & Exponential & 169.59 (10.309) & 192.56 (5.322) & 1.82 (0.042) & 2.85 (0.029) \\
Non-normal & 5 & 20 & Tri-diagonal & 151.21 (9.872) & 199.70 (5.491) & 1.53 (0.032) & 2.55 (0.025) \\
Non-normal & 5 & 20 & Exponential & 145.62 (9.728) & 200.69 (5.775) & 1.86 (0.043) & 2.86 (0.030) \\
\bottomrule
\end{tabular}}
\end{table}

The DFLIM procedure consistently outperforms the MEWMA and MGLR procedures across all shift patterns. 
The MEWMA procedure detects shifts only in a few cases but with significant delay. MEWMA employs a profile monitoring technique, where the matrix is flattened into a long vector and segmented for separate handling. This approach risks losing spatial correlations due to flattening and segmentation. Another drawback is that each segment remains high-dimensional, making covariance estimation for each segment challenging. In our experiment, with 800 in-control data and segments of dimension 200, MEWMA's performance suffers due to poor marginal covariance estimation across segments. Additionally, MEWMA assumes temporal independence, which is invalid in our complex auto-correlations experiments.

The MGLR procedure reduces data dimensionality by defining ROIs, but this process can result in local information loss within each ROI.
Considering the nature of our shifts (sparse or alternating positive and negative), the mean of entries in each ROI tends to be close to zero. This cancellation of informative entries by taking the mean of each ROI in MGLR undermines its ability to detect a shift effectively.

The ST-SSD procedure is the most competitive baseline compared to DFLIM. In Table \ref{tab:sparse arl1}, DFLIM detects sparse shifts faster than ST-SSD, saving approximately 70 observations. Regarding the ring shift in Table \ref{tab:ring arl1}, under normal noise distribution, DFLIM performs slightly worse than ST-SSD, with a lag of less than 10 observations. However, with non-normal noise distribution, ST-SSD outperforms DFLIM by approximately 20 observations, although both achieve significantly smaller $\ARL_1$ values compared to MEWMA and MGLR. In Table \ref{tab:sine arl1}, both DFLIM and ST-SSD detect the chessboard shift almost instantly, showing negligible differences between them. For the chessboard shift in Table \ref{tab:chessboard arl1}, both DFLIM and ST-SSD procedures detect the shift nearly instantly, with a negligible difference between them.

The ST-SSD procedure employs a regression framework to decompose observations into three components: the in-control mean matrix $M_0$, the shift $A$, and the noise. Except for the sparse shift, ST-SSD achieves this decomposition effectively, resulting in successful detection. DFLIM incorporates static and dynamic statistics to construct monitoring statistics for change detection. In scenarios where a shift is algebraically similar to the in-control mean matrix, static features play a significant role in change detection, as seen in the chessboard shift of Figure \ref{fig:chessboard}. On the other hand, dynamic features dominate when a shift is algebraically different from the in-control mean matrix, as demonstrated in the sparse, ring, and sine shifts of Figures \ref{fig:sparse}-\ref{fig:sine}. Hence, DFLIM robustly and efficiently detects changes across various settings.

As stated in Theorem~\ref{thm:mean difference of nu}, the first-type features $\beta_{{i},t}$ consistently help achieve change detection when $M_0$ and $M_1$ are algebraically similar. This is empirically supported by the effectiveness of  DFLIM in detecting the chessboard shift, which resembles the in-control mean. On the other hand, the effectiveness of the second-type features $\gamma_{{i},t}$ becomes evident when dealing with large matrix dimensions, approaching the asymptotic theory outlined in Theorem~\ref{thm:mean difference of gamma}. Experimental results show that normally distributed noises often lead to smaller $\ARL_1$ compared to non-normal noises, likely due to slower convergence to the asymptotic theory associated with non-normal noises.

\begin{figure}[b!]
\centering
\resizebox*{0.8\linewidth}{!}{\includegraphics{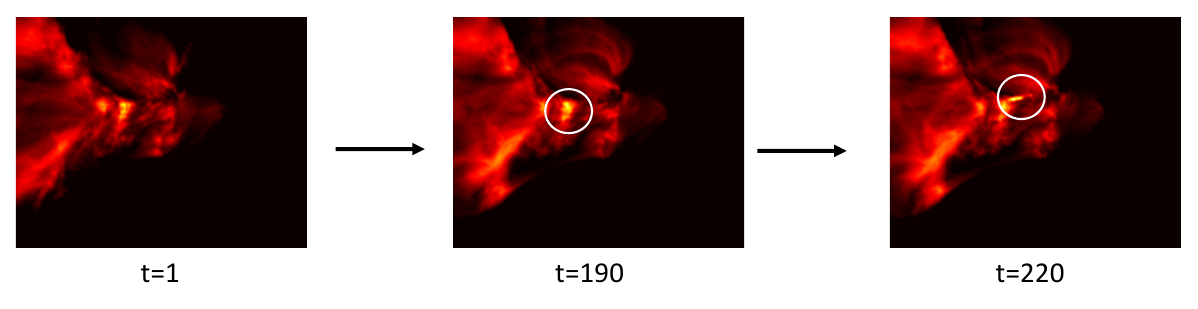}}
\caption{Solar flare images at $t=1,190,220$. The white circles mark outbursts.}
\label{fig:solarshift}
\end{figure}
% \begin{figure}[b!]
% \centering 
% \begin{subfigure}[h]{0.8\linewidth}
% \includegraphics[width=\linewidth]{plots/solarburst.pdf}
% \end{subfigure}
% \caption{The solar flare image sequences at $t=1,190,220$. The white circles mark the outbursts.}
% \label{fig:solarshift}
% \end{figure}
\begin{figure}[b!]
\centering
\resizebox*{0.85\linewidth}{!}{\includegraphics{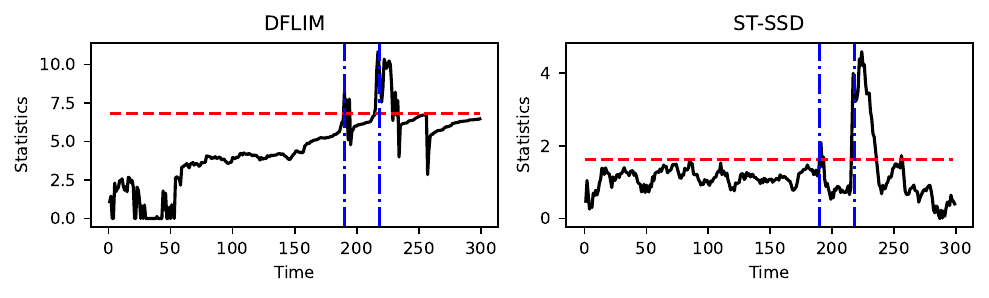}}
\caption{The monitoring process of DFLIM and ST-SSD on the solar flare images. Two outbursts occur around times 190 and 220, respectively, marked by vertical blue dashed lines. In each procedure, the control limit is indicated by a red dashed line, with the monitoring statistics shown in black.}
\label{fig:flare}
\end{figure}
% \begin{figure}[b!]
% \centering 
% \begin{subfigure}[h]{0.85\linewidth}
% \includegraphics[width=\linewidth]{plots/real1_realsolar.pdf}
% \end{subfigure}
% \caption{The monitoring process of DFLIM and ST-SSD on the solar flare images. Two outbursts happen around time 190 and 220, respectively, marked by vertical blue dashed lines. For each procedure, the control limit is marked by a red dashed line. The monitoring statistics are in black.}
% \label{fig:flare}
% \end{figure}

\section{Real Data Experiments}\label{sec:real}
In this section, we apply DFLIM to real datasets to illustrate its broad applicability. More specifically, we analyze solar flare images in Section~\ref{subsec:solar} and stochastic textured surface images in Section~\ref{subsec:texture}. We compare DFLIM with ST-SSD, excluding MEWMA and MGLR, because their control limits cannot be determined with a single in-control sequence. 

We determine the control limit of DFLIM analytically by solving equation \eqref{formula:control limit} for $H$.
For ST-SSD, we determine its control limit using an empirical $(1-1/\ARL_0)$ quantile estimate of the in-control monitoring statistics \cite[Section~5.2]{yan2018real}.

\subsection{Solar flare outburst}
\label{subsec:solar}
In this example, we aim to detect solar flare outbursts. Figure \ref{fig:solarshift} shows solar flare images at times $t=1, 190, 220$. The solar flare outbursts are represented by the bright spots in the images, indicated by the white circles. Prior knowledge indicates that two outbursts happen around times 190 and 220, respectively. The first outburst around time 190 is relatively moderate, while the second one around time 220 is more intense. The sequence consists of 300 images, each represented by a $232$-by-$292$ matrix. We use the first 150 matrices as the training dataset and then perform monitoring on the entire sequence. 

Detecting solar flare outbursts presents several challenges. First, each image is high-dimensional, containing nearly 70,000 pixels. Second, the low-rank property is not inherently applicable to solar flare images. To address this issue, we employ a patch technique that breaks the data into {\it patches} to promote the low-rank structure \citep{lowe1999object}. Each patched image exhibits a numerical rank of $10$. Third, the dynamics of the changes are complex due to multiple change points and slowly evolving backgrounds. Specifically, many time points around 190 and 220 experience outbursts. To handle these multiple outbursts, we restart monitoring once an alarm is raised. Additionally, the changes in the data involve not only intense outbursts but also slow shifts in the background. To address the dynamic background, we process the data by taking consecutive differences of the images after the patch technique is performed, and we set the target $\ARL_0 = 50,000$. 

Figure \ref{fig:flare} shows the results of ST-SSD and DFLIM to the solar flare images. 
ST-SSD effectively detects both moderate and intense
outbursts without triggering false alarms. However, we observe that the monitoring statistic of ST-SSD is very close to its control limit around $t=70, 80, 100$, corresponding to periods of normal solar flare activity. DFLIM demonstrates superior performance compared to ST-SSD, effectively identifying the outbursts around $t = 190$ and $t=220$ with the detection statistics away from control limits prior to $t=190$.

\subsection{Stochastic textured surface monitoring}
\label{subsec:texture}

\begin{figure}[b!]
\centering
\resizebox*{0.8\linewidth}{!}{\includegraphics{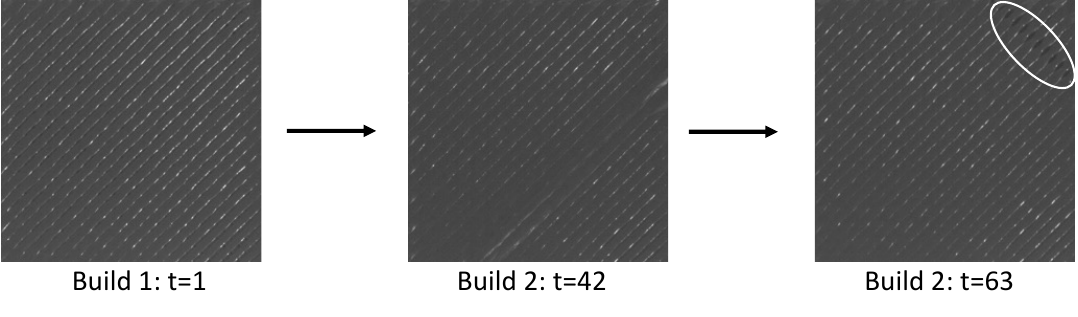}}
\caption{The material extrusion process is depicted at  $t=1,42,63$. At $t=42$, there is a transition from the first to the second build. The bead orientation is $45^\circ$, and the white ellipse marks a defect.}
\label{fig:textured-monitor}
\end{figure}
% \begin{figure}[b!]
% \centering 
% \begin{subfigure}[h]{0.75\linewidth}
% \includegraphics[width=\linewidth]{plots/texturedmonitor.pdf}
% \end{subfigure}
% \caption{The process of material extrusion at index $t=1,42,63$. At $t=42$, the build shift from the first to the second. Here is bead orientation is $45^\circ$. The white ellipse marks the defects.}
% \label{fig:textured-monitor}
% \end{figure}
\begin{figure}[b!]
\centering
\resizebox*{0.85\linewidth}{!}{\includegraphics{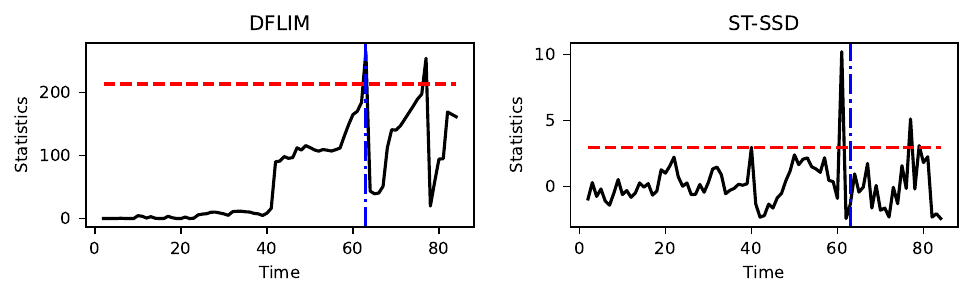}}
\caption{The monitoring process of DFLIM and ST-SSD on the 3D printing textured surface images with bead orientation at $45^\circ$. A build-to-build transition happens at $t=42$, while a defect is introduced at $t=63$, marked by vertical blue dashed lines in both figures. In each procedure, the control limit is indicated by a red dashed line, with the monitoring statistics shown in black.}
\label{fig:textured_monitor_45}
\end{figure}
% \begin{figure}[b!]
% \centering 
% \begin{subfigure}[h]{0.85\linewidth}
% \includegraphics[width=\linewidth]{plots/realtexture45_ARL0_4000.pdf}
% \end{subfigure}

% \caption{The monitoring process of DFLIM and ST-SSD on the 3D printing textured surface images for the bead orientation at $45^\circ$. The build-to-build transition happens at $t=42$. The defect is introduced at $t=63$, marked by vertical blue dashed lines in both figures. For each procedure, the control limit is marked by a red dashed line. The monitoring statistics are in black.}
% \label{fig:textured_monitor_45}
% \end{figure}
\begin{figure}[t!]
\centering
\resizebox*{0.85\linewidth}{!}{\includegraphics{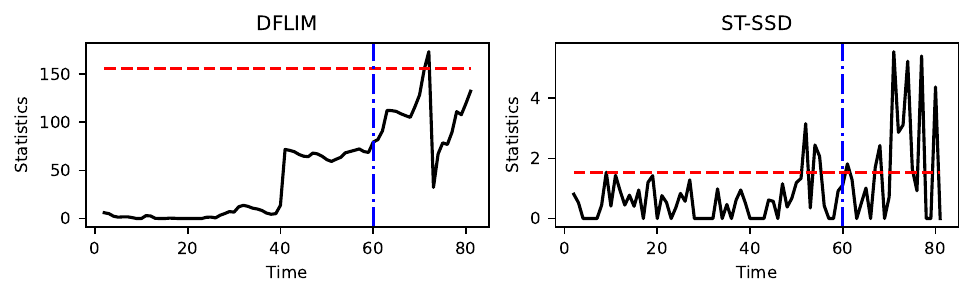}}
\caption{The monitoring process of DFLIM and ST-SSD on the 3D printing textured surface images with bead orientation at $135^\circ$. A build-to-build transition happens at $t=41$, while a defect is introduced at $t=60$, marked by vertical blue dashed lines in both figures. In each procedure, the control limit is indicated by a red dashed line, with the monitoring statistics shown in black.}
\label{fig:textured_monitor_135}
\end{figure}
% \begin{figure}[t!]
% \centering 
% \begin{subfigure}[h]{0.85\linewidth}
% \includegraphics[width=\linewidth]{plots/realtexture135_ARL0_4000.pdf}
% \end{subfigure}
% \caption{The monitoring process of DFLIM and ST-SSD on the 3D printing textured surface images for bead orientation at $135^\circ$. The build-to-build transition happens at $t=41$. The defect is introduced at $t=60$, marked by vertical blue dashed lines in both figures. For each procedure, the control limit is marked by a red dashed line. The monitoring statistics are in black.}
% \label{fig:textured_monitor_135}
% \end{figure}

The online monitoring of additive manufacturing processes, commonly referred to as 3D printing, has drawn increasing attention due to its potential to reduce material waste. This example involves monitoring the production process of a parallelepiped ($20 \times 20 \times 20$ mm)  using fused filament fabrication. 

\cite{Caltanissetta2023} use data from a sequential process to print two parallelepipeds. We use the same dataset and, refer to them as Build 1 (representing the process of building the first parallelepiped) and Build 2 (representing the process of building the second parallelepiped) hereafter. Build 1 is in control, while defects are intentionally introduced into Build 2 in the middle of the printing process.

During the printing process, layerwise images are captured by a video-imaging system installed above the printing area. To optimize bonding between consecutive layers, it is recommended that the material extrusion direction be rotated iteratively. Consequently, the captured images are categorized into two types based on bead orientations: $45^\circ$ and $135^\circ$. Each build consists of two sequences labeled as $45^\circ$ and $135^\circ$, respectively, and these sequences are treated separately. Despite both sequences originating from the same build process, some images are occasionally skipped and not captured. Therefore, the index $t$ in this example corresponds to the index of captured images and does not directly translate into time. The sequence with bead orientation $45^\circ$ ($135^\circ$) consists of $84$ ($81$) images. For Build 1 (Build 2), the sequence ends at $t=42$ ($t=41$), after which defects are introduced at $t=63$ ($t=60$). Each image is represented as a $250$-by-$250$ matrix. During monitoring, we utilize images from Build 1 to set up control limits and implementation parameters, then apply DFLIM and ST-SSD to the all the images. 

Unlike the solar flare images, we do not need the patch technique for this dataset because the in-control data naturally exhibits low-rank properties due to the aligned paths of material extrusions. However, this dataset presents challenges similar to those in the solar flare images, including high dimensionality and evolving backgrounds. Additionally, two more challenges arise with this dataset. First, the training sample size is small, consisting of only about 40 images. Second, the dataset exhibits build-to-build variability, which stems from dynamic factors in the printing area, such as changing illumination conditions during transitions between builds. Technically, the shift between builds could be considered a change point, but it is undesirable to detect this inter-build shift. Instead, we aim to detect a shift caused by actual defects, but the build-to-build variability increases the risk of false alarms. 

For monitoring, we still take consecutive differences of images and restart the process upon detecting a change point.  We set $\ARL_0 = 4000$ for both bead orientations, suggesting approximately half a month between consecutive printer overhauls. Figure~\ref{fig:textured_monitor_45} displays the monitoring process for bead orientation $45^\circ$ while Figure~\ref{fig:textured_monitor_135} shows the case for bead orientation $135^\circ$. In both cases, ST-SSD detects the change at the true change point but raises false alarms. 
DFLIM immediately detects the defect without raising a false alarm at the build-to-build transition in Figure~\ref{fig:textured_monitor_45}. DFLIM still detects the defect in Figure~\ref{fig:textured_monitor_135} but has a delay of $12$ images, roughly equivalent to one minute in real time (without considering skipped images).

\section{Conclusion}\label{sec:conclusion}
In this paper, we propose a distribution-free monitoring procedure named DFLIM that can address the challenges posed by modern image data, including complex spatial-temporal dependence, non-normality, and high dimensionality with low-dimensional structure. We provide a comprehensive theoretical discussion on the detection ability and the behavior of  $\ARL_0$ and $\ARL_1$ for the proposed procedure under reasonable assumptions. Extensive simulations are conducted using various distributions, ranks, and spatial-temporal correlation structures to validate the generality of DFLIM. Additionally, we apply DFLIM to two real datasets, solar flare datasets, and additive manufacturing datasets, to demonstrate its applicability to real-world scenarios.

A future direction is to extend the method to approximately low-rank matrices. While Assumption~\ref{assump:low_rank} assumes a low-rank in-control mean matrix \(M_0\), in practice, its estimate \(\widehat{M}_0\) often deviates from this structure. Significant noise in \(\widehat{M}_0\) may lead to overestimating the rank \(r\), introducing extra projected components in \(y_t\) and increasing false negatives or false alarms. Analyzing the extent of this deviation could guide a more appropriate selection of \(r\), enhancing DFLIM's performance.

\section*{Acknowledgement}
This work is partially supported by an NSF CAREER CCF-1650913, NSF DMS-2134037, CMMI-2015787, CMMI-2112533, DMS-1938106, DMS-1830210, and the Coca-Cola Foundation.

\bibliographystyle{plain}
\bibliography{main}

\appendix
\section{Overlapping weighted Cram\'er-von Mises (CvM) estimator}\label{sec:CvM}
\begin{algorithm}[t!]
\caption{The overlapping weighted CvM estimator}
\label{algo:CvM}
\textbf{Input:} In-control monitoring statistics  $\{T_t:t=1,\ldots,n\}$, batch size $m$.
\begin{algorithmic}[1] 
\State For $i=1,\ldots, n-m+1$, calculate
$$
C_i = \frac{1}{m}\sum_{j=1}^m g\left(\frac{j}{m}\right)\frac{j^2}{m}\left(\bar{T}_{i,j}-\bar{T}_i \right)^2, 
$$
where the function $g(t) = -24+150t-150t^2$, the partial batch mean $\bar{T}_{i,j} = \frac{1}{j}\sum_{j'=1}^j T_{i+j'}$ and the batch mean $\bar{T}_{i} = \frac{1}{m}\sum_{j=1}^m T_{i+j}$.

\State The CvM estimator is the average over all CvM estimators from the  batches, namely 
$$
\Omega_0^2 = \frac{1}{n-m+1}\sum_{i=1}^{n-m+1} C_i.
$$

 \end{algorithmic}
\end{algorithm}

The CvM estimator is proposed by \cite{alexopoulos2007overlapping}. An expedient strategy for determining the batch size $m$ could ensure approximate independence among the batch means. Any applicable statistical test for independence can be employed for this purpose. 
n our study, we adopt a fixed batch size of $m=50$ across all experimental configurations. For the purpose of self-containedness, we show the computation process of CvM estimator in Algorithm~\ref{algo:CvM}.

\section{Proofs} \label{app:proof}

\begin{proof}[Proof of Theorem~\ref{thm:mean difference of gamma}]
    We begin by studying the statistical behaviors of in-control and out-of-control $\gamma_{{i},t}$ separately. 
\begin{itemize}
    \item For the in-control case, Lemma \ref{lem:in-control gamma} implies that $\gamma_{{i},t} = \operatorname{o}(\sqrt{p_1p_2})$ almost surely. 
    \item For an out-of-control case, Lemma \ref{lem:out-of-control gamma} implies 
    $$
    \gamma_{{i},t} = \sqrt{p_1p_2}\bar{\rho}_{i} + z_{i}+ \operatorname{o}\left(\sqrt{p_1p_2}\right)+\operatorname{o}_p\left(1\right),
    $$
\end{itemize}
where $\operatorname{o}_p\left(1\right)$ represents a random variable $Z$ satisfying $\lim_{p_1,p_2\to\infty} \bbP \left(|Z|\ge\epsilon\right)=0$ for any positive constant $\epsilon$. Then, we have 
$$
\frac{\E_1 [\gamma_{{i},t}]-\E_0 [\gamma_{{i},t}]}{\sqrt{p_1p_2}} = \bar{\rho}_{i} + \operatorname{o}\left(1\right).
$$
\end{proof}

\begin{proof}[Proof of Theorem~\ref{lem:beyond iid gain in T square}]
Without specifying the probability measure, we can decompose $T_t$ into the following components. 
    \begin{equation*}
        \begin{aligned}
            &\E_\upsilon \sqbrack{T_t} = \E_\upsilon \sqbrack{\left( y_t - \E_0 [y_t] \right)^\top  \cov_0^{-1}(y_t) \left( y_t - \E_0[y_t] \right)}\\
            & = \E_\upsilon \sqbrack{\left( y_t - \E_\upsilon [y_t] + \E_\upsilon [y_t] - \E_0 [y_t] \right)^\top \cov_0^{-1}(y_t) \left( y_t - \E_\upsilon [y_t] + \E_\upsilon [y_t] - \E_0 [y_t]\right)}\\
            & = \tr\cirbrack{\cov_0^{-1}(y_t) \cov_\upsilon^{-1}(y_t) }
             + {\left( \E_\upsilon [y_t] - \E_0 [y_t] \right)^\top  \cov_0^{-1}(y_t) \left(\E_\upsilon [y_t] - \E_0 [y_t] \right)}, \quad \upsilon = 0,1. 
        \end{aligned}
    \end{equation*}
    Now we examine the difference between the in-control and out-of-control expectations:
    \begin{equation*}
        \begin{aligned}
            \E_1\sqbrack{T_t} - \E_0\sqbrack{T_t} &=
            \tr\cirbrack{\cov_0^{-1}(y_t) \cov_1^{-1}(y_t) }
             \\
             &\quad + {\left( \E_1 [y_t] - \E_0 [y_t] \right)^\top  \cov_0^{-1}(y_t) \left(\E_1 [y_t] - \E_0 [y_t] \right)}-2r\\
             &= \tr\cirbrack{\Sigma^{-1} \tilSig  }
             + \delta^\top  \Sigma^{-1} \delta -2r.
        \end{aligned}
    \end{equation*}
    By utilizing the property of the Schur complement, we can express the inverse of $\Sigma$ in the following form:
    \begin{equation*}
        \begin{aligned}
            \Sigma^{-1} & 
            = \begin{bmatrix}
            \Sigma_\beta^{-1}+\Sigma_\beta^{-1}P\left(\Sigma/\Sigma_\beta\right)^{-1}P^\top\Sigma_{\beta}^{-1} & 
             -\Sigma_\beta^{-1}P \left(\Sigma/\Sigma_\beta\right)^{-1} \\[6pt]
            -\left(\Sigma/\Sigma_\beta\right)^{-1}P^\top \Sigma_\beta^{-1}  & 
            \left(\Sigma/\Sigma_\beta\right)^{-1}
            \end{bmatrix},
        \end{aligned}   
    \end{equation*}
    where $\Sigma/\Sigma_\beta = \Sigma_\gamma-P^\top \Sigma_\beta^{-1} P $ denotes the Schur complement of the block $\Sigma_\beta$ within the matrix $\Sigma$.     For $\tr\cirbrack{\Sigma^{-1} \tilSig  }$, we have
    $$
    \begin{aligned}
        &\Sigma^{-1}\tilSig = 
            \begin{bmatrix}
            \Sigma_\beta^{-1}+\Sigma_\beta^{-1}P\left(\Sigma/\Sigma_\beta\right)^{-1}P^\top\Sigma_{\beta}^{-1} & 
             -\Sigma_\beta^{-1}P \left(\Sigma/\Sigma_\beta\right)^{-1} \\[6pt]
            -\left(\Sigma/\Sigma_\beta\right)^{-1}P^\top \Sigma_\beta^{-1}  & 
            \left(\Sigma/\Sigma_\beta\right)^{-1}
            \end{bmatrix} 
            \begin{bmatrix}
            \Sigma_{\beta} & \tilP \\[6pt]
            \tilP^\top & \tilSig_{\gamma} \\
            \end{bmatrix}\\[6pt]
            &\,= 
            \resizebox{\linewidth}{!}{
            $\begin{bmatrix}
            I_r+\Sigma_\beta^{-1}P\left(\Sigma/\Sigma_\beta\right)^{-1}P^\top-\Sigma_\beta^{-1}P \left(\Sigma/\Sigma_\beta\right)^{-1}\tilP^\top & 
             \left(\Sigma_\beta^{-1}+\Sigma_\beta^{-1}P\left(\Sigma/\Sigma_\beta\right)^{-1}P^\top\Sigma_{\beta}^{-1}\right)\tilP-\Sigma_\beta^{-1}P \left(\Sigma/\Sigma_\beta\right)^{-1}\tilSig_{\gamma}\\[6pt]
             -\left(\Sigma/\Sigma_\beta\right)^{-1}P^\top+\left(\Sigma/\Sigma_\beta\right)^{-1}\tilP^\top & 
            -\left(\Sigma/\Sigma_\beta\right)^{-1}P^\top \Sigma_\beta^{-1}\tilP+\left(\Sigma/\Sigma_\beta\right)^{-1}\tilSig_\gamma
            \end{bmatrix}$
            }.
    \end{aligned}
    $$
    Then, its trace becomes
    $$
    \begin{aligned}
        &\tr\cirbrack{\Sigma^{-1} \tilSig} \\
        &= r+\tr\left(\Sigma_\beta^{-1}P\left(\Sigma/\Sigma_\beta\right)^{-1}P^\top-\Sigma_\beta^{-1}P \left(\Sigma/\Sigma_\beta\right)^{-1}\tilP^\top-\left(\Sigma/\Sigma_\beta\right)^{-1}P^\top \Sigma_\beta^{-1}\tilP+\left(\Sigma/\Sigma_\beta\right)^{-1}\tilSig_\gamma\right)\\
        & = r+\tr\left[\left(\Sigma/\Sigma_\beta\right)^{-1}\left(
        P^\top\Sigma_\beta^{-1}P-
       \tilP^\top\Sigma_\beta^{-1}P -
        P^\top \Sigma_\beta^{-1}\tilP+
       \tilSig_\gamma\right)\right]\\
       & = r+\tr\left[\left(\Sigma/\Sigma_\beta\right)^{-1}\left(
        P^\top\Sigma_\beta^{-1}P-
       \tilP^\top\Sigma_\beta^{-1}P -
        P^\top \Sigma_\beta^{-1}\tilP+
        \tilP^\top \Sigma_\beta^{-1}\tilP-
        \tilP^\top \Sigma_\beta^{-1}\tilP+
       \tilSig_\gamma\right)\right]\\
          & = r+\tr\left\{\left(\Sigma/\Sigma_\beta\right)^{-1}\left\{\left[
        \Sigma_{\beta}^{-1/2}\left(P-\tilP\right)\right]^\top\left[
        \Sigma_{\beta}^{-1/2}\left(P-\tilP\right)\right]+
       \tilSig/\Sigma_{\beta}\right\}\right\}\\
        & = r+\tr\left\{\left(\Sigma/\Sigma_\beta\right)^{-1}\left[
        \Sigma_{\beta}^{-1/2}\left(P-\tilP\right)\right]^\top\left[
        \Sigma_{\beta}^{-1/2}\left(P-\tilP\right)\right]\right\}+\tr\left[\left(\Sigma/\Sigma_\beta\right)^{-1}
       \left(\tilSig/\Sigma_{\beta}\right)\right],
    \end{aligned}
    $$
    where $\tilSig/\Sigma_\beta = \tilSig_\gamma-\tilP^\top \Sigma_\beta^{-1} \tilP$. 
    For $\delta^\top\Sigma^{-1}\delta$, we have
    $$
    \begin{aligned}
        &\delta^\top\Sigma^{-1}\delta \\
        &= 
        \begin{bmatrix}
            \delta_\beta^\top & \delta_\gamma^\top
        \end{bmatrix}
        \begin{bmatrix}
            \Sigma_\beta^{-1}+\Sigma_\beta^{-1}P\left(\Sigma/\Sigma_\beta\right)^{-1}P^\top\Sigma_{\beta}^{-1} & 
             -\Sigma_\beta^{-1}P \left(\Sigma/\Sigma_\beta\right)^{-1} \\[6pt]
            -\left(\Sigma/\Sigma_\beta\right)^{-1}P^\top \Sigma_\beta^{-1}  & 
            \left(\Sigma/\Sigma_\beta\right)^{-1}
        \end{bmatrix}
        \begin{bmatrix}
            \delta_\beta \\[6pt]
            \delta_\gamma
        \end{bmatrix}\\[6pt]     &\,=\delta_\beta^\top\left[\Sigma_\beta^{-1}+\Sigma_\beta^{-1}P\left(\Sigma/\Sigma_\beta\right)^{-1}P^\top\Sigma_{\beta}^{-1}\right]\delta_\beta\\[6pt]
        &\,\quad -\delta_\beta^\top\Sigma_\beta^{-1}P \left(\Sigma/\Sigma_\beta\right)^{-1}\delta_\gamma -\delta_\gamma^\top\left(\Sigma/\Sigma_\beta\right)^{-1}P^\top \Sigma_\beta^{-1}\delta_\beta + \delta_\gamma^\top\left(\Sigma/\Sigma_\beta\right)^{-1}\delta_\gamma\\[6pt]
        &\,=\left\|\Sigma_\beta^{-1/2}\delta_\beta\right\|^2 + \left\|\left(\Sigma/\Sigma_\beta\right)^{-1/2}\left(P^\top\Sigma_{\beta}^{-1}\delta_\beta+\delta_\gamma\right)\right\|^2.
    \end{aligned}
    $$

    Now, we substitute $\tr\cirbrack{\Sigma^{-1} \tilSig  }$ and $\delta^\top\Sigma^{-1}\delta$ back into the expression for the shift size of $T_t$. Finally, we have
    $$
    \begin{aligned}
        &\E_1\sqbrack{T_t} - \E_0\sqbrack{T_t} 
         \\
         &= \tr\cirbrack{\Sigma^{-1} \tilSig  }
         + \delta^\top  \Sigma^{-1} \delta -2r\\
         &\,=\tr\left[\left(\Sigma/\Sigma_\beta\right)^{-1}
       \left(\tilSig/\Sigma_{\beta}\right)\right]-r+\tr\left\{\left(\Sigma/\Sigma_\beta\right)^{-1}\left[
        \Sigma_{\beta}^{-1/2}\left(P-\tilP\right)\right]^\top\left[
        \Sigma_{\beta}^{-1/2}\left(P-\tilP\right)\right]\right\}\\ 
       &\, \quad + \left\|\Sigma_\beta^{-1/2}\delta_\beta\right\|^2 + \left\|\left(\Sigma/\Sigma_\beta\right)^{-1/2}\left(P^\top\Sigma_{\beta}^{-1}\delta_\beta+\delta_\gamma\right)\right\|^2.
    \end{aligned}
    $$
\end{proof}

\end{document}